\documentclass[orivec]{llncs}
\usepackage{amsmath}
\usepackage{amssymb}
\usepackage{amsfonts}
\usepackage[misc,geometry]{ifsym} 
\usepackage{graphicx}
\DeclareGraphicsExtensions{.pdf,.png,.jpg}
\usepackage{float} 
\usepackage{tikz}
\usetikzlibrary{petri}
\usetikzlibrary{decorations,decorations.pathreplacing,decorations.pathmorphing,decorations.markings}
\usetikzlibrary{arrows}
\usetikzlibrary{shapes.geometric}
\usetikzlibrary{shapes.symbols}
\usepackage{tkz-graph}
\usepackage[T1]{fontenc}
\usepackage{mathbbol}
\usepackage[ruled,vlined,linesnumbered]{algorithm2e}
\DeclareSymbolFontAlphabet{\mathbbl}{bbold}
\usepackage{mathtools}

\begin{document}

\renewenvironment{proof}{{\em Proof.}}{\qed} 

\newcommand{\cstate}[1]{\ensuremath{{}_{{\color{black!70}\lfloor}#1}}} 
\newcommand{\mysyst}{\mathcal{S}} 
\newcommand{\lts}{LTS}

\newcommand{\es}{\emptyset}
\newcommand{\pminus}{
\mbox{\textrm{$-$}\!\!\!\!\!\!\!\:\,\,\raisebox{1.5mm}{$\scriptstyle \bullet$}}\:}
\newcommand{\dt}{\bullet}
\newcommand{\nsymbol}{\mathbb{N}}
\newcommand{\zsymbol}{\mathbb{Z}}
\newcommand{\qsymbol}{\mathbb{Q}}
\newcommand{\rsymbol}{\mathbb{R}}
\newcommand{\support}{\mathit{supp}}
\newcommand{\emptyseq}{\varepsilon}
\newcommand{\Parikh}{{\mathbf P}}

\def\projection#1#2{\mathchoice
              {\setbox1\hbox{${\displaystyle #1}_{\scriptstyle #2}$}
              \projectionaux{#1}{#2}}
              {\setbox1\hbox{${\textstyle #1}_{\scriptstyle #2}$}
              \projectionaux{#1}{#2}}
              {\setbox1\hbox{${\scriptstyle #1}_{\scriptscriptstyle #2}$}
              \projectionaux{#1}{#2}}
              {\setbox1\hbox{${\scriptscriptstyle #1}_{\scriptscriptstyle #2}$}
              \projectionaux{#1}{#2}}}
\def\projectionaux#1#2{{#1\,\smash{\vrule height .8\ht1 depth .85\dp1}}_{\,#2}} 

\newcommand{\is}{\iota}
\newcommand{\uniqueP}{\Upsilon}
\newcommand{\zero}{\mathbbl{0}}
\newcommand{\zeroT}{\mathbbl{0}^{\mid T \mid}} 
\newcommand{\oneP}{\mathbbl{1}^{\mid P \mid}}

\newcommand{\TS}{\mathit{TS}}

\itemsep0pt

\tikzstyle{mypetristyle}=[
      place/.style={circle, draw=black, fill=white, thick, minimum size=3mm},
      transition/.style={rectangle, draw=black, fill=black!8, thick, minimum height=2mm, minimum width=2mm,inner sep=2pt},
      token/.style={circle,draw=black,fill=black,inner sep=0pt,minimum size=1mm}
]

\tikzset{every picture/.style={mypetristyle}}
\tikzset{every label/.style={black!90}}

\title{Efficient Synthesis of Weighted Marked Graphs with Circular Reachability Graph, and Beyond} 
\author{Raymond Devillers\inst{1} \and Evgeny Erofeev(\Letter)\thanks{Supported by DFG through \mbox{grant Be 1267/16-1} {\tt ASYST}.}\inst{2} \and Thomas Hujsa\thanks{Supported by the STAE foundation/project DAEDALUS, Toulouse, France.}\inst{3}}
\institute{
D\'epartement d'Informatique, Universit\'e Libre de Bruxelles,\\ B-1050 Brussels, Belgium
(\email{rdevil@ulb.ac.be})
\and
Department of Computing Science, Carl von Ossietzky Universit\"at Oldenburg,\\ D-26111 Oldenburg, Germany 
(\email{evgeny.erofeev@informatik.uni-oldenburg.de})
\and
LAAS-CNRS, Universit\'e de Toulouse, CNRS, Toulouse, France 
(\email{thujsa@laas.fr})
}

\maketitle

\begin{abstract} 
In previous studies, several methods have been developed to synthesise Petri nets from labelled transition systems (\lts{}), 
often with structural constraints on the net and on the \lts{}.
In this paper, we focus on Weighted Marked Graphs (WMGs) and Choice-Free (CF) Petri nets,
two weighted subclasses of nets in which each place has at most one output;
WMGs have the additional constraint that each place has at most one input.

We provide new conditions for checking the existence of a WMG whose reachability graph is isomorphic to a given circular \lts{}, 
i.e. forming a single cycle;
we develop two new polynomial-time synthesis algorithms dedicated to these constraints:
the first one is LTS-based (classical synthesis) while the second one is vector-based (weak synthesis) and more efficient in general.
We show that our conditions also apply to CF synthesis in the case of three-letter alphabets,
and we discuss the difficulties in extending them to CF synthesis over arbitrary alphabets.

\end{abstract}


\section{Introduction}\label{intro.sec}

Petri nets form a highly expressive and intuitive operational model of discrete event systems,
capturing the mechanisms of synchronisation, conflict and concurrency. 
Many of their fundamental behavioural properties are decidable, 
allowing to model and analyse numerous artificial and natural systems.
However,
most interesting model checking problems are worst-case intractable, 
and the efficiency of synthesis algorithms varies widely depending 
on the constraints imposed on the desired solution.
In this study,
we focus on the Petri net synthesis problem from a labelled transition system (\lts{}),
which consists in determining the existence of a Petri net whose reachability graph is isomorphic to the given \lts{},
and building such a Petri net solution when it exists.

In previous studies on analysis or synthesis, structural restrictions on nets 
encompassed \emph{plain} nets (each weight equals $1$; also called ordinary nets) \cite{murata89},
\emph{homogeneous} nets (for each place $p$, all the output weights of $p$ are equal) \cite{STECS,HD2017},
\emph{free-choice} nets (the net is plain, and any two transitions sharing an input have the same set of inputs) \cite{DesEsp,STECS},  
{\em join-free} nets (each transition has at most one input place) \cite{STECS,ACSD13,TECS14,HD2017}.
Recently, another kind of restriction has been considered, limiting the number of distinct labels of the \lts{} 
\cite{BarylskaBEMP15,BarylskaBEMP16,ErofeevBMP16,ErofeevW17}.

Depending on the constraints on the solution to be constructed, the complexity of the synthesis problem can vary widely:
the problem can be solved in polynomial-time for bounded Petri nets~\cite{PolyTimeSynthesis95},
while aiming at \emph{elementary net systems}, or at various other Petri net subclasses with fixed marking bound, 
makes the problem NP-complete~\cite{EnsNPcomplete97,NPcompleteSynthesisClasses2019}.

In this paper, we study the solvability of \lts{} with weighted marked graphs (WMGs; each place has at most one output and one input) 
and choice-free nets (CF; each place has at most one output). 
Both classes are important for real-world applications, and are widely studied in the 
literature~\cite{tcs97,PN14,DEH18,chep,WTS92,TCS17,DH2018,BestDSW18}.
We focus mainly on finite \emph{circular \lts{}}, 
i.e. strongly connected \lts{} that contain a unique \emph{cycle}\footnote{A set $A$ of $k$ arcs in a \lts{} $G$ defines a cycle of $G$ 
if the elements of $A$ can be ordered as a sequence $a_1 \ldots a_k$
such that, for each $i \in \{1, \ldots, k\}$, $a_i = (n_i,\ell_i,n_{i+1})$ and
$n_{k+1} = n_1$, i.e. the $i$-th arc $a_i$ goes from node $n_i$ to node $n_{i+1}$ until the first node $n_1$ is reached, closing the path. 
}.
In this context, we investigate the \emph{cyclic solvability} of a word $w$, 
meaning the existence of a Petri net solution to the finite circular \lts{} induced by the infinite \emph{cyclic word} $w^\infty$.
These restrictions appear in practical situations,
since various complex applications can be decomposed into subsystems satisfying such constraints 
\cite{phdhujsa,BarylskaBEMP15,ErofeevBMP16,besdev-ccc25,TCS17,DS18,RD-acta18,articul19}.\\

\noindent {\bf Contributions.} 
We study further the links between simple \lts{} structures and the reachability graph of WMGs and CF nets, as follows. 
First, we show that a binary (i.e. over a two-letter alphabet) \lts{} is CF-solvable if and only if it is WMG-solvable. 
Then, we develop new conditions for WMG-solving a cyclic word over an arbitrary alphabet, 
with a polynomial-time synthesis algorithm. 

We show that a word over a three-letter alphabet is cyclically WMG-solvable iff it is cyclically CF-solvable,
and that this result does not hold with four-letter alphabets.
More generally, we discuss the difficulties of extending these results to CF synthesis over arbitrary alphabets.

We introduce the notion of {\em weak synthesis}, which aims at synthesising a Petri net from a given transition-vector $\uniqueP$ 
instead of a sequence: 
the solution obtained enables some sequence whose Parikh vector equals $\uniqueP$.
This allows to be less restrictive on the solution design.
Then, we provide a polynomial-time algorithm for the weak synthesis of WMGs with circular reachability graphs.

Finally, we show that our weak synthesis algorithm
performs generally much faster than the sequence-based algorithm.
 
Comparing with \cite{DEH-ATAED19}, we provide more details,
we add the equivalence result on CF nets for three-letter alphabets
in Subsection~\ref{nocf.subsec} and the new Section~\ref{WeakSynthesis} on weak synthesis,
with a new synthesis algorithm and the study of its complexity.\\

\noindent {\bf Organisation of the paper.} 
After recalling classical definitions, notations and properties in Section \ref{Def.sec},
we present the equivalence of CF- and WMG-solvability for $2$-letter words in Section \ref{twoletters.sec}.

In Section \ref{kletters.sec}, we focus on circular \lts{}:
we give a new characterisation of WMG-solvability and a dedicated polynomial-time synthesis algorithm.  
We prove the equivalence between cyclic WMG and CF synthesis for three-letter alphabets.
We also provide a number of examples showing that some of our results 
cannot be applied to the class of CF-nets over arbitrary alphabets.

Section~\ref{WeakSynthesis} contains our study of the weak synthesis problem for WMGs with a circular reachability graph,
with a new polynomial-time synthesis algorithm.
Finally, Section \ref{conclu.sec} presents our conclusions and perspectives.

\section{Classical Definitions, Notations and Properties}\label{Def.sec}

{\bf \lts{}, sequences and reachability.} 
A {\em labelled transition system with initial state}, {\em \lts{}} for short, 
is a quadruple $\TS=(S,\to,T,\is)$ where $S$ is the set of {\em states}, 
$T$ is the (finite) set of {\em labels}, 
$\to\,\subseteq(S\times T\times S)$ is the {\em transition relation}, 
and 
$\is\in S$ is the {\em initial state}.
A label $t$ is {\em enabled} at $s\in S$, written $s[t\rangle$, if $\exists s'\in S\colon(s,t,s')\in\to$, 
in which case $s'$ is said to be {\em reachable} from $s$ by the firing of $t$, and we write $s[t\rangle s'$.
Generalising to any (firing) sequences $\sigma\in T^*$,
$s[\emptyseq\rangle$ and $s[\emptyseq\rangle s$ are always true, with $\emptyseq$ being the empty sequence;
and $s[\sigma t\rangle s'$, i.e.,  $\sigma t$ is {\em enabled} from state $s$ and leads to $s'$ 
if there is some $s''$ 
with $s[\sigma\rangle s''$ and $s''[t\rangle s'$.
For clarity,
in case of long formulas we write $\cstate{r} \sigma \cstate{s} \tau \cstate{q}$ instead of $r [\sigma \rangle s [\tau \rangle q$, 
thus fixing some intermediate states along a firing sequence. 
A state $s'$ is {\em reachable} from state $s$ if $\exists\sigma\in T^*\colon s[\sigma\rangle s'$.
The set of states reachable from $s$ is noted $[s\rangle$.\\

\noindent {\bf Petri nets and reachability graphs.} 
A (finite, place-transition) \emph{weighted Petri net}, or \emph{weighted net},
is a tuple $N=(P,T,W)$ where
$P$ is a finite set of {\em places}, 
$T$ is a finite set of {\em transitions}, with $P\cap T=\es$
and
$W$ is a {\em weight} function $W\colon((P\times T)\cup(T\times P))\to\nsymbol$ giving the weight of each arc.
A \emph{Petri net system}, or \emph{system}, is a tuple $\mathcal S=(N,M_0)$ where 
$N$ is a net and $M_0$ is the {\em initial marking},
which is a mapping $M_0\colon P\to\nsymbol$ (hence a member of $\nsymbol^P$)
indicating the initial number of {\em tokens} in each place. 
The {\em incidence matrix} $I$ of the net is the integer $P\times T$-matrix with components $I(p,t)=W(t,p)-W(p,t)$.

A place $p \in P$ is {\em enabled by} a marking $M$ if $M(p) \ge W(p,t)$ for every transition $t\in T$. 
A transition $t\in T$ is {\em enabled by} a marking $M$, 
denoted by $M[t\rangle$, if for all places $p\in P$, $M(p)\geq W(p,t)$.
If $t$ is enabled at $M$, then $t$ can {\em occur} (or {\em fire}) in $M$, 
leading to the marking $M'$ defined by $M'(p)=M(p)-W(p,t)+W(t,p)$; we note $M[t\rangle M'$.
A marking $M'$ is {\em reachable} from $M$ if there is a sequence of firings leading from $M$ to $M'$.
The set of markings reachable from $M$ is denoted by $[M\rangle$.
The {\em reachability graph of $\mathcal S$} is the labelled transition system $\mathit{RG}(\mathcal S)$ 
with the set of vertices $[M_0\rangle$, the set of labels $T$, initial state $M_0$
and transitions $\{(M,t,M')\mid M,M'\in[M_0\rangle\land M[t\rangle M'\}$.
A system $\mathcal S$ is {\em bounded} if $\mathit{RG}(\mathcal S)$ is finite.\\

\noindent {\bf Vectors.} The {\em support} of a vector is the set of the indices of its non-null components.
Consider any net $N=(P,T,W)$ with its incidence matrix~$I$.
A {\em T-vector} (respectively {\em P-vector}) is an element of $\nsymbol^{T}$ (respectively $\nsymbol^{P}$); 
it is called {\em prime} if the greatest common divisor of its components is one 
(i.e., it is non-null and its components do not have a common non-unit factor). 
A {\em T-semiflow} $\nu$ of the net is a non-null T-vector 
such that $I\cdot\nu=\zero$. 
A T-semiflow is called {\em minimal} when it is prime and its support 
is not a proper superset of the support of any other T-semiflow~\cite{tcs97}.

The {\em Parikh vector} $\Parikh(\sigma)$ of a finite transition sequence $\sigma$ is a T-vector 
counting the number of occurrences of each transition in $\sigma$,
and the {\em support} of $\sigma$ is the support of its Parikh vector, 
i.e., $\support(\sigma)=\support(\Parikh(\sigma))=\{t\in T\mid\Parikh(\sigma)(t)>0\}$.\\

\noindent {\bf Strong connectedness and cycles in \lts{}.} 
The \lts{} is said {\em reversible} if, 
$\forall s\in[\is\rangle$, we have $\is\in[s\rangle$, i.e., it is always possible to go back to the initial state;
reversibility implies the strong connectedness of the \lts{}.

A sequence $s[\sigma\rangle s'$ is a {\em cycle}, or more precisely a {\em cycle at (or around) state $s$}, if $s=s'$.
A non-empty cycle $s[\sigma\rangle s$ is called {\em small} if there is no non-empty cycle 
$s'[\sigma'\rangle s'$ in $\TS$ with $\Parikh(\sigma') \lneqq \Parikh(\sigma)$ 
(the definition of Parikh vectors extends readily to sequences over the set of labels $T$ of the \lts{}). 
A cycle $s[\sigma\rangle s$ is {\em prime} if $\Parikh(\sigma)$ is prime. 
$\TS$ has the {\em prime cycle property} if each small cycle has a prime Parikh vector.

A \emph{circular \lts{}} is a finite, strongly connected \lts{} 
that contains a unique cycle; hence, it has the shape of an oriented circle.
The circular \lts{} \emph{induced by} a word $w\!=\!w_1\ldots w_k$ is defined as $s_0[w_1\rangle s_1 [w_2\rangle s_2 \ldots [w_k\rangle s_0$
with initial state $s_0$.\\
All notions defined for labelled transition systems apply to Petri nets through their reachability graphs.\\

\noindent {\bf Petri net subclasses.} A net $N=(P,T,W)$ is {\em plain} if no arc weight exceeds $1$;  
{\em pure} if $\forall p\in P\colon(p^\dt{\cap}{}^\dt p)=\es$, 
where $p^\dt=\{t\in T\mid W(p,t){>}0\}$ and ${}^\dt p=\{t\in T\mid W(t,p){>}0\}$;
{\it choice-free} (CF) \cite{crespi-mandrioli-75,tcs97} or place-output-nonbranching (ON) \cite{besdev-ccc25}
if $\forall p\in P\colon|p^\dt|\leq1$; 
a {\em weighted marked graph} (WMG) if $|p^\dt|\leq 1$ and $|{}^\dt p|\leq 1$ for all places $p\in P$.
The WMGs form a subclass of the CF nets and contain the weighted T-systems (WTSs) of~\cite{WTS92}, 
also known as {\em weighted event graphs} (WEGs) in \cite{March09}, in which $\forall p \in P$,
$|{}^\dt p| = 1$ and $|p^\dt| = 1$. 
Plain WEGs are also known as {\em marked graphs}~\cite{chep} or {\em T-nets}~\cite{DesEsp}.\\

\noindent {\bf Isomorphism and solvability.} 
Two \lts{} $\TS_1=(S_1,\to_1,T,s_{01})$ and $\TS_2=(S_2,\to_2,T,s_{02})$ are
isomorphic if there is a bijection $\zeta\colon S_1\to S_2$ with $\zeta(s_{01})=s_{02}$ and 
$(s,t,s')\in\to_1\,\Leftrightarrow(\zeta(s),t,\zeta(s'))\in\to_2$, for all $s,s'\in S_1$.
If an \lts{} $\TS$ is isomorphic to $\mathit{RG}(\mathcal S)$, where $\mathcal S$ is a system,
we say that $\mathcal S$ {\em solves} $\TS$.
Solving a word $w=\ell_1 \ldots \ell_k$ amounts to
solve the acyclic \lts{} defined by the single path $\is [\ell_1\rangle s_1 \ldots [\ell_k\rangle s_k$. 
A finite word $w$ is \emph{cyclically solvable} if the circular \lts{} induced by $w$ is solvable.
An \lts{} is WMG- (or CF-)solvable if a WMG (or a CF system) solves it.\\

\noindent {\bf Separation problems.}
Let $TS=(S,\to,T,s_0)$ be a given labelled transition system.
The theory of regions~\cite{bbd} characterises the solvability of an \lts{} through the solvability of a set of {\em separation problems}.
In case the \lts{} is finite, we have to solve
$\frac{1}{2}{\cdot}|S|{\cdot}(|S|{-}1)$ states separation problems 
and up to $|S|{\cdot}|T|$ event/state separation problems, as follows:
\begin{itemize}
\item A {\em region} of $(S,\to,T,s_0)$ is a triple 
$(\mathbb{R},\mathbb{B},\mathbb{F}) \in (S \to \mathbb{N}, T \to \mathbb{N}, T\to \mathbb{N})$ 
such that for all $s[t\rangle s'$, $\mathbb{R}(s) \ge \mathbb{B}(t)$ and 
$\mathbb{R}(s') = \mathbb{R}(s) - \mathbb{B}(t) + \mathbb{F}(t)$. 
A region models a place $p$, in the sense that $\mathbb{B}(t)$ models $W(p,t)$, $\mathbb{F}(t)$ models $W(t,p)$, 
and $\mathbb{R}(s)$ models the token count of $p$ at the marking corresponding to $s$.

\item A {\em states separation problem} (SSP for short) consists of a set of states $\{s,s'\}$ with $s\neq s'$,
and it can be solved by a place distinguishing them, 
i.e., has a different
number of tokens in the markings corresponding to the two states.
\item An {\em event/state separation problem} (ESSP for short) consists of a pair $(s,t)\in S{\times}T$ with $\lnot s[t\rangle$.
For every such problem, one needs a place $p$ such that $M(p)<W(p,t)$ 
for the marking $M$ corresponding to state $s$, where $W$ refers to the arcs of the hoped-for net.
On the other hand, for every edge $(s',t, s'') \in \to$ we must guarantee $M'(p)\ge W(p,t)$, $M'$ being
the marking corresponding to state $s'$.
\end{itemize}
If the \lts{} is infinite, also the number of separation problems (of each kind) becomes infinite.

A synthesis procedure does not necessarily lead to a connected solution.
However, the technique of decomposition into prime factors described in \cite{dev-ACSD16,RD-acta18}
can always be applied first, so as to handle connected partial solutions and recombine them afterwards. 
Hence, in the sequel, we focus on connected nets, w.l.o.g. 
In the next section, we consider the CF synthesis problem with two distinct labels.

\section{Reversible Binary CF Synthesis}\label{twoletters.sec}
 
In this section, we relate CF- to WMG-solvability for binary reversible \lts{}. 
\begin{lemma}[Pure CF-solvability]\label{pure-CF.lem}\\
If a reversible \lts{} has a CF solution, it has a pure CF solution. 
\end{lemma}
\begin{proof}
Let $\TS=(S,\to,T,\is)$ be a reversible \lts{}. 
If $t\in T$ does not occur in $\to$, 
$\TS$ is solvable iff $\TS'=(S,\to,T\setminus\{t\},\is)$ is solvable
and a possible solution of $\TS$ is obtained by adding to any solution of $\TS'$ a transition $t$ and 
a fresh place $p$, initially empty, with an arc from $p$ to $t$ (e.g. with weight~1), so that 
$p$ is not a side condition\footnote{A place $p$ is a {\em side condition} if ${}^\dt p \cap p^\dt \neq \emptyset$. }. 
We can thus assume that each label of $T$ occurs in $\to$. 


\begin{figure}[!ht]  
\centering
\begin{tikzpicture}[scale=0.6]
\node[draw,place,inner sep=1pt](p)at(3,0.5){\scriptsize $\mu_0$}; 
\node[below of=p,node distance=20pt]{$p$}; 
\node[transition] (x) at (6,0.5) {$x$}; 
\node[transition] (a1) at (0,1.8) {$a_1$};
\node[transition] (a2) at (0,0.9) {$a_2$};
\node(vd)at(0,0.3){\Large ${\vdots}$};
\node[draw,transition] (am) at (0,-0.6) {$a_m$};
\draw(p)[]edge[-latex,out=-20,in=205]node[below,swap,inner sep=4pt,pos=0.45]{$k{+}h$}(x);
\draw(x)[]edge[-latex,out=155,in=20]node[above,swap,inner sep=4pt,pos=0.55]{$h$}(p);
\draw(a1)[]edge[-latex]node[auto,inner sep=1pt,pos=0.2]{$k_1$}(p);
\draw(a2)[]edge[-latex]node[auto,inner sep=1pt,pos=0.2]{$k_2$}(p);
\draw(am)[]edge[-latex]node[below,inner sep=4pt,pos=0.4]{$k_m$}(p);
\end{tikzpicture}
\caption{A general pure ($h=0$) or non-pure ($h>0$) choice-free place $p$ with initial marking $\mu_0$.
Place $p$ has at most one outgoing transition named $x$.
The set $\{a_1,\ldots,a_m\}$ comprises all other transitions, i.e., $T=\{x,a_1,\ldots,a_m\}$,
and $k_j$ denotes the weight of the arc from $a_j$ to $p$ (which could be zero).
}
\label{CF-place.fig}
\end{figure}
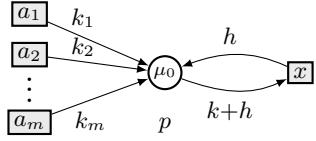
The general form of a place in a CF solution is exhibited in Fig.~\ref{CF-place.fig}.
If $h=0$, we are done, so that we shall assume $h>0$.
If $-h\leq k< 0$, the marking of $p$ cannot decrease, 
and since $x$ occurs in $\to$, the system cannot be reversible.
If $k=0$, for the same reason all the $k_i$'s must be null too, $\mu_0\geq h$, and we may drop $p$.
Hence we assume that $k>0$ and $\exists i:k_i>0$.

Once $x$ occurs, the marking of $p$ is at least $h$, remains so, and since the system is reversible, 
all the reachable markings have at least $h$ tokens in $p$. 
But then, if we replace $p$ by a place $p'$ with initially $\mu_0-h$ tokens, the same $k_i$'s and $h=0$, 
we get exactly the same reachability graph, 
but with $h$ tokens less in $p'$ than in $p$. 
This will wipe out the side condition for $p$, and repeating this for each side condition, 
we get an equivalent pure and choice-free solution.
\end{proof}

\begin{theorem}[Reversible binary CF-solvability]\label{CF=WMG.thm}\\ 
A binary reversible \lts{} is CF-solvable iff it is WMG-solvable.
\end{theorem}
\begin{proof}
If we have two labels, from Lemma~\ref{pure-CF.lem}, if there is a CF solution, 
there will be one with places of the form exhibited in
Fig.~\ref{solcyc.fig}, hence a WMG solution.
\end{proof}

\begin{figure}[!ht]
\centering
\begin{tikzpicture}[scale=0.9]
\node[place,inner sep=1pt](p)at(2,0)[label=above:$p_{a,b}$]{\scriptsize $\mu_0$}; 
\node[transition](a)at(0,0){$a$};
\node[transition](b)at(4,0){$b$};
\draw(a)[]edge[-latex]node[above left]{$m$}(p);
\draw(p)[]edge[-latex]node[above right]{$n$}(b);
\end{tikzpicture}
\caption{A generic pure CF-place with two labels.
}
\label{solcyc.fig}
\end{figure}
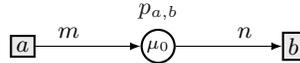

In the next section, the number of letters is no more restricted.

\section{Cyclic WMG- and CF-solvablity}\label{kletters.sec}

In this section,
we recall and extend conditions for WMG-solvability of some restricted classes of \lts{} formed by a single circuit,
which were developed in~\cite{DEH18}. 

We gradually study the separation problems -- SSPs in Subsection~\ref{ssp.subsec} and ESSPs in Subsection~\ref{essp.subsec} -- for 
cyclic solvability with WMGs, leading to a language-theoretical characterisation of cyclically WMG-solvable sequences.  
The characterisation gives rise to a polynomial-time synthesis algorithm in Subsection~\ref{firstalgo.subsec}, 
which is shown to be more efficient than the classical synthesis approach. 

Finally, in Subsection~\ref{nocf.subsec}, we study the extensibility of these results to the CF case: 
for three-letter alphabets,
we show that a word is cyclically WMG-solvable iff it is cyclically CF-solvable; 
unfortunately, for arbitrary alphabets, we show with the help of examples that the other results cannot be directly extended.

In the following, two distinct labels $a$ and $b$ are called {\em (circularly) adjacent} in a word $w$
if $w = (w_1 ab w_2)$ or $w = (b w_3 a)$ for some ${w_1,w_2,w_3\in T^*}$.
We denote by $p_{a,\ast}$ any place $p_{a,b}$ where $b$ is adjacent to $a$.
Also, if $T=\{t_0,t_1,\ldots,t_m\}$ with $m>0$, 
at least one label is adjacent to $t_0$, 
and at each point at least one label is adjacent to the ones we distinguished so far, 
until we get the whole set $T$;
we can thus start from any label $t_i$ instead of $t_0$.

\begin{theorem}[Sufficient condition for cyclic WMG-solvability~\cite{DEH18}]\label{kary.theo}\\ 
Consider any word $w$ over any finite alphabet $T$ such that $\Parikh(w)$ is prime.
Suppose the following:
$\forall u = \projection{w}{t_1 t_2}$ 
(i.e., the projection\footnote{The projection of a word $w\in A^*$ 
on a set $A' \subseteq A$ of labels is the maximum subword of $w$ 
whose labels belong to $A'$, noted $\projection{w}{A'}$. 
For example, the projection of the word $w = \ell_1 \, \ell_2 \, \ell_3 \, \ell_2$ 
on the set $\{\ell_1 ,\, \ell_2\}$ is the word $\ell_1 \, \ell_2 \, \ell_2$.} 
 of $w$ on $\{t_1,t_2\}$)
for some circularly adjacent labels $t_1,t_2$ in $w$, 
$u = v^\ell$ for some positive integer $\ell$,
$\Parikh(v)$ is prime,
and $v$ is cyclically solvable by a circuit (i.e., a circular net).
Then,
$w$ is cyclically solvable with a WMG.
\end{theorem}

\begin{theorem}[Cyclic WMG-solvability of ternary words~\cite{DEH18}]\label{cyclicsolvternary.theo}\\ 
Consider a ternary word $w$ (with three letters in its alphabet $T$) 
with Parikh vector $(x,x,y)$ such that $\gcd(x,y)=1$. 
Then,
$w$ is cyclically solvable with a WMG
if and only if, for any pair $t_1\neq t_2 \in T$ such that
$w = (w_1 t_1 t_2 w_2)$ or $w = (t_2 w_3 t_1)$, 
$u = v^\ell$ for some positive integer $\ell$ with $u = \projection{w}{t_1t_2}$,
$\Parikh(v)$ is prime,
and $v$ is cyclically solvable by a circuit.
\end{theorem}

\begin{figure}[!ht] 
\centering
\begin{tikzpicture}[scale=0.63]
\begin{scope}[xshift=-4cm,yshift=0cm]
\node[place,tokens=0](p0)at(2,2.)[]{$ $};
\node[place,tokens=1](p1)at(0,3.2)[]{$ $};
\node[place,tokens=4](p2)at(0,0.755)[]{$ $};
\node[place,tokens=0](p3)at(4,0.75)[]{$ $};
\node[place,tokens=0](p4)at(4,3.2)[]{$ $};
\node[transition](a)at(0,2.){$a$}
	edge[post]node[left]{$ $}(p1)
	edge[post]node[above]{$ $}(p0)
	edge[pre]node[right]{$2$}(p2)
	;
\node[transition](b)at(4,2.){$b$}
	edge[pre]node[above]{$ $}(p0)
	edge[post]node[left]{$2$}(p3)
	edge[pre]node[below right]{$ $}(p4)
	;
\node[transition](c)at(2,3.2){$c$}
	edge[pre]node[below]{$3$}(p1)
	edge[post]node[below]{$3$}(p4)
	;
\node[transition](d)at(2,0.75){$d$}
	edge[pre]node[above]{$3$}(p3)
	edge[post]node[above]{$3$}(p2)
	;
\end{scope}
\begin{scope}[xshift=1.7cm,yshift=0.5cm]
\node[place,tokens=1](p0)at(5,0.2)[]{$ $};
\node[place,tokens=0](p1)at(5,3)[]{$ $};
\node[place,tokens=2](p2)at(0,0.2)[]{$ $};
\node[place,tokens=1](p3)at(0,3)[]{$ $};
\node[place,tokens=0](p4)at(2.5,1.9)[]{$ $};
\node[place,tokens=3](p5)at(1,1)[]{$ $};
\node[place,tokens=0](p6)at(4,1)[]{$ $};
\node[transition](a)at(0,1.9){$a$}
	edge[post]node[above]{$ $}(p4)
	edge[post]node[above]{$ $}(p3)
	edge[pre]node[above]{$ $}(p5)
	edge[pre]node[left]{$ $}(p2)
	;
\node[transition](b)at(5,1.9){$b$}
	edge[pre]node[]{$ $}(p1)
	edge[pre]node[above]{$ $}(p4)
	edge[post]node[above, near end]{$ $}(p6)
	edge[post]node[]{$ $}(p0)
	;
\node[transition](c)at(2.5,3){$c$}
	edge[pre]node[below]{$3$}(p3)
	edge[post]node[below]{$3$}(p1)
	;
\node[transition](d)at(2.5,1){$d$}
	edge[pre]node[above]{$3$}(p6)
	edge[post]node[above]{$3$}(p5)
	;
\node[transition](e)at(2.5,0.2){$e$}
	edge[post]node[above, near start]{$3$}(p2)
	edge[pre]node[above, near start]{$3$}(p0)
	;
\end{scope}
\begin{scope}[xshift=8.5cm,yshift=0cm]
\node[place,tokens=0](p0)at(2,2.)[]{$ $};
\node[place,tokens=1](p1)at(0,3.2)[]{$ $};
\node[place,tokens=4](p2)at(0,0.755)[]{$ $};
\node[place,tokens=0](p3)at(4,0.75)[]{$ $};
\node[place,tokens=2](p4)at(4,3.2)[]{$ $};
\node[transition](a)at(0,2.){$a$}
	edge[pre]node[left]{$ $}(p1)
	edge[post]node[above]{$ $}(p0)
	edge[pre]node[right]{$2$}(p2)
	;
\node[transition](b)at(4,2.){$b$}
	edge[pre]node[above]{$ $}(p0)
	edge[post]node[left]{$2$}(p3)
	edge[post]node[below right]{$ $}(p4)
	;
\node[transition](c)at(2,3.2){$c$}
	edge[post]node[below]{$3$}(p1)
	edge[pre]node[below]{$3$}(p4)
	;
\node[transition](d)at(2,0.75){$d$}
	edge[pre]node[above]{$3$}(p3)
	edge[post]node[above]{$3$}(p2)
	;
\end{scope}
\end{tikzpicture}
\caption{The WMG on the left solves $aacbbdabd$
cyclically, and the WMG in the middle solves 
$aacbbeabd$ cyclically.
On the right, the WMG solves $abcabdabd$ cyclically. 
}
\label{counterex.fig}
\end{figure}
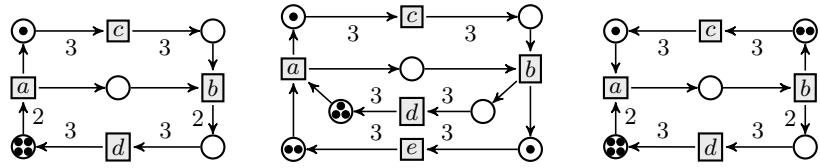

For a circular \lts{}, the solvability of its binary projections by circuits is a sufficient condition, as specified by 
 Theorem~\ref{kary.theo}, 
but it turns out not to be a necessary one. 
Indeed, for the cyclically solvable sequence $w_1 = aacbbdabd$ 
(cf. left of Fig.~\ref{counterex.fig}), 
its binary projection on $\{ a,b\}$ is $\projection{w_1}{a,b} = aabbab$ which is not cyclically solvable with a WMG 
(neither generally solvable).
Looking only at the Parikh vector of the sequence is also not enough to establish its cyclic (un)solvability. 
For instance, sequences $w_2 = abcabdabd$ and $w_3 = abcbadabd$ 
are Parikh-equivalent: $\Parikh(w_2) = \Parikh(w_3) = (3,3,1,2)$ 
(and also Parikh-equivalent to $w_1$), but $w_2$ is cyclically solvable with a WMG 
(e.g. with the WMG on the right of Fig.~\ref{counterex.fig}) and $w_3$ is not WMG-cyclically solvable.

All the binary projections of $w_1$ and $w_3$ are cyclically WMG-solvable, except $\projection{w_i}{a,b}$. 
Only the unsolvability of $\projection{w_3}{a,b}$ implies the unsolvability of $w_3$. 
Since all the $w_i$ are Parikh-equivalent, then so are their binary projections.
Thus, we have to analyse the sequences themselves, without abstracting to Parikh vectors. 
Since the projections $\projection{w_1}{a,b}$ and $\projection{w_3}{a,b}$ are equivalent 
(up to cyclic rotation and swapping $a$ and $b$), 
it is not sufficient to check the `problematic' binary projections. 
We then study the conditions for solvability of separation problems.

\subsection{SSPs for Prime Cycles}\label{ssp.subsec} 

For any word $w=t_0\ldots t_k$, for $0 \le i,j \le k$ such that $i\neq j$, we note 
$\Parikh_{ij} = \Parikh(t_i t_{i+1} \ldots t_{j-1} )$ if $i < j$ and 
$\Parikh_{ij} = \Parikh(t_i t_{i+1} \ldots t_{k-1} t_k t_0 t_1 \ldots t_{j-1} )$ if $i > j$. 

\begin{lemma}[SSPs are solvable for prime cycles]\label{SSP.lem}  
For a cyclic transition system $TS = (S,\to, T, s_0)$ defined by some word $w=t_0\ldots t_k$,  
where $S= \{s_0, \ldots, s_k\}$, $\to = \{(s_{i-1}, t_{i-1}, s_{i})\mid 1\le i \le k\}\cup \{(s_k,t_k,s_0)\}$, 
if $\Parikh(w)$ is prime then all the SSPs are solvable.
\end{lemma}
\begin{proof}
If $|T|=1$, then $k=0$  and $|S|=1$, so that there is no SSP to solve. We may thus assume $|T|>1$.

For each pair of distinct labels $a,b \in T$ that are adjacent in $TS$, 
construct places $p_{a,b}$ (and $p_{b,a}$ since adjacency is commutative) 
as in Fig.~\ref{solcyc.fig} with 
\begin{equation}\label{mn.eq}
	m = \frac{\Parikh(w)(b)}{\gcd(\Parikh(w)(a),\Parikh(w)(b))},\;
	n = \frac{\Parikh(w)(a)}{\gcd(\Parikh(w)(a),\Parikh(w)(b))}, 
\end{equation}
and $\mu_0 = n\cdot \Parikh(w)(b)$.
Clearly, the markings of $p_{a,b}$ reachable by repeatedly firing $u= \projection{w}{ab}$ 
are always non-negative, and the initial marking is reproduced after each repetition of the sequence $u$. 
Consider two distinct states $s_i, s_j \in S$ (w.l.o.g. $i<j$).
We now demonstrate that there is at least one place of the form $p_{a,b}$ such that $M_i(p_{a,b}) \neq M_j(p_{a,b})$, 
where $M_l$ denotes the marking corresponding to state $s_l$ for $0 \le l\le k$. 
If $j-i=1$, then any place of the form $p_{t_{i},\ast}$ distinguishes states $s_i$ and $s_j$. 
The same is true if $j-i >1$ but $\forall l\in[i,j-1]:t_l=t_i$. 
Otherwise, choose some letter $a$ from $t_{i} \ldots t_{j-1}$ and an adjacent letter $b$.
Then $M_j(p_{a,b}) = M_i(p_{a,b}) + m\cdot \Parikh_{ij}(a) - n\cdot \Parikh_{ij}(b)$. 
If $M_i(p_{a,b})\neq M_j(p_{a,b})$, place $p_{a,b}$ distinguishes $s_i$ and $s_j$.
Otherwise we have $m\cdot \Parikh_{ij}(a) = n\cdot \Parikh_{ij}(b)$, hence, due to the choice of $m$ and $n$:
	\[ \frac{\Parikh_{ij}(a)}{\Parikh(w)(a)} = \frac{\Parikh_{ij}(b)}{\Parikh(w)(b)}
	\]
	(so that $b$ also belongs to $t_{i} \ldots t_{j-1}$).
Consider some other letter $c$ which is adjacent to $a$ or $b$. If place $p_{a,c}$  distinguishes $s_i$ and $s_j$, we are done. 
Otherwise, due to the choice of the arc weights for these places, we have 
	\[ \frac{\Parikh_{ij}(a)}{\Parikh(w)(a)} = \frac{\Parikh_{ij}(c)}{\Parikh(w)(c)} = \frac{\Parikh_{ij}(b)}{\Parikh(w)(b)}. 
	\]
Since $t_i \ldots t_{j-1}$ is finite, by progressing along the adjacency relation, 
either we find a place which has different markings at $s_i$ and $s_j$, 
or for all $a,b \in \support(t_i \ldots t_{j-1})$ we have  
	\[
		\frac{\Parikh_{ij}(a)}{\Parikh(w)(a)} = \frac{\Parikh_{ij}(b)}{\Parikh(w)(b)}. 
	\]
If $\support(t_i \ldots t_{j-1}) = \support(w)$, 
$\Parikh(w)$ is proportional to $\Parikh(t_i \ldots t_{j-1})$, but since $t_i \ldots t_{j-1}$ is smaller than $w$
(otherwise $s_i=s_j$) this contradicts the primality of $\Parikh(w)$. 
Hence, there exist adjacent $c$ and $d$ such that $c \in \support(w)\setminus \support(t_i \ldots t_{j-1})$ 
and $d \in \support(t_i \ldots t_{j-1})$. For the place $p_{c,d}$ we have $M_j(p_{c,d}) \neq M_i(p_{c,d})$.
\end{proof}

\vspace{\baselineskip} 
This property has some similarities with Theorem 4.1 in \cite{DH-FI19}, but the preconditions are different.
The reachability graph of any CF net, hence of any WMG, satisfies the prime cycle property~\cite{BD-I-and-C,bds17}. 
Thus, primeness of a sequence avoids solving SSPs when aiming at these two classes of Petri nets.

\subsection{ESSPs in Cyclic WMG-Solvability}\label{essp.subsec}

Now, we develop further conditions for the cyclic WMG-solvability. 

\begin{lemma}[Special form of WMG solutions for cycles]\label{adjacent.lem}
If $w \in T^\ast$ is cyclically solvable by a WMG, there exists a WMG $\mysyst= ((P,T,W),M_0)$, 
where $P$ consists of places $p_{a,b}$, for each pair of distinct 
circularly adjacent $a$ and $b$ (i.e., either  $w=u_1abu_2$ or $w=bu_3a$).
\end{lemma}
\begin{proof}
Consider a sequence $w = t_0 \ldots t_k$, where $\Parikh(w)$ is prime. 
Let us assume that the system $((P,T,W),M_0)$ is a WMG solving $w$ cyclically. 
Due to the definition of WMGs, all the places that we have to consider are of the form schematised in Fig.~\ref{genplaceab.fig}. 
The arc weights may differ due to the parameter $l > 0$, 
but the ratio $\frac{W(a, p_{a,b})}{W(p_{a,b},b)} = \frac{m}{n}$ is determined by the Parikh vector of $w$ 
and its cyclic solvability; the initial marking is to be defined.
Moreover, we have to consider only those places which are connected to the pairs of circularly adjacent transitions in $w$. 
Indeed, if $w = u_1  \cstate{s_i} a \cstate{s_{i+1}} b\; u_2$, where $b\neq a$, $s_i$ is the state reached after performing $u_1$ 
and $s_{i+1}$ is the state reached after performing $u_1 a$, then 
any place that solves the ESSP $\lnot M_{i}[b \rangle$ is an input place for $b$. 
On the other hand, any place whose marking at $M_{s_i}$ differs from its marking at $M_{s_{i+1}}$ is connected to $a$. 
Hence, a place $p \in P$ solving $\lnot M_{i}[b \rangle$ is of the form $p_{a, b}$. 
Since $p$ is only affected by $a$ and $b$, 
it also disables $b$ at all the states between $s_{l}$ and $s_i$ in $w$
when it is of the form 
$w = u_3  \cstate{s_j} t_{j} \cstate{s_{j+1}} b^+ \cstate{s_l} u_4 \cstate{s_i} a b u_2$ with $\Parikh(u_4)(b) = 0$ 
(in the case there is no $b$ in the prefix between $s_0$ and $abu_2$, $s_l=s_0$). 
Analogously, if  $t_{j} \neq b$, there must be a place $q \in P$ 
of the form $p_{t_j,b}$ that solves $\lnot M_{s_j} [b \rangle$. 
Doing so, we ascertain that the places of the form schematised
in Fig.~\ref{genplaceab.fig} for the adjacent pairs of transitions are sufficient to handle all the ESSPs.

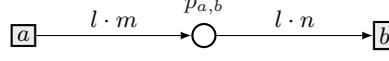
\begin{figure}[!ht] 
\centering
\begin{tikzpicture}[scale=0.8]
\begin{scope}[xshift=0cm,yshift=0cm]
\node[place,tokens=0](p)at(2,0)[label=above:$p_{a,b}$]{$  $};
\node[transition](a)at(-1,0){$a$};
\node[transition](b)at(5,0){$b$};
\draw(a)[]edge[-latex]node[above]{$l\cdot m$}(p); 
\draw(p)[]edge[-latex]node[above]{$l\cdot n$}(b);
\end{scope}
\end{tikzpicture}
\vspace*{-1mm}
\caption{A general place from $a$ to $b$ in a WMG solution of $w$: $m = \Parikh(w)(b)$, $n = \Parikh(w)(a)$, $l$ may be any multiple of $1/\gcd(m,n)$.} 
\label{genplaceab.fig}
\end{figure}

In fact, for each pair of adjacent transitions $a$ and $b$ in $w$, a single place of the form $p_{a,b}$ is sufficient. 
Indeed, assume there are $p_1, p_2 \in P$ of the form $p_{a,b}$. 
If $\frac{M_0(p_1)}{\gcd(W(a,p_1),W(p_1,b))} \ge \frac{M_0(p_2)}{\gcd(W(a,p_2),W(p_2,b))}$
then for any $M \in [M_0\rangle$, $M(p_1) < W(p_1,b)$ implies $M(p_2) < W(p_2,b)$. 
Hence, $p_1$ is redundant in the system. 
It means we can choose $l$ as we want (among the multiples of $1/\gcd(m,n)$) 
and only keep the place of the form $p_{a,b}$ in any solution with the smallest initial marking. 
Note that it may happen that we need a place $p_{a,b}$, but not $p_{b,a}$.
\end{proof}

\vspace*{\baselineskip}
The existence of a WMG solution of this special form 
allows us to establish a necessary condition for the cyclic solvability of sequences. 

\begin{lemma}[A necessary condition for cyclic solvability with a WMG]\label{nec.lem}
If $w\in T^\ast$ is cyclically solvable by a WMG, then for any adjacent transitions $a$ and $b$ in $w$, 
and any two occurrences of $ab$ in  
	$ w = u_1 \, \cstate{s_r} a \, b \, \ldots \, \cstate{s_q} \, a \, b \, u_2, 
	$
the inequality
\begin{equation}\label{interval.eq} 
	 \frac{\Parikh_{rj}(b)-1}{\Parikh_{rj}(a)} < \frac{m}{n} < \frac{\Parikh_{jq}(b)+1}{\Parikh_{jq}(a)} 
\end{equation}
holds true, where $m,n$ are as in~(\ref{mn.eq}), $ r\le j \le q$, and
the right inequality is omitted when $\Parikh_{jq}(a) = 0$ and the left inequality is omitted when $\Parikh_{rj}(a) = 0$. 
\end{lemma}
\begin{proof}
Let $w$ be cyclically solvable with a WMG $\mysyst=((P,T,W),M_0)$ 
as in Lemma~\ref{adjacent.lem}, and place $p \in P$ 
be of the form $p_{a,b}$ (as in Fig.~\ref{genplaceab.fig}, with $l=1$ and a well chosen initial marking) 
for an adjacent pair $ab$.
Choose two $ab$'s in 
$w = u_1 \, \cstate{s_r} a \, \cstate{s_{r+1}} \, b \, \cstate{s_{r+2}} \, \ldots \, \cstate{s_q} \, a \, \cstate{s_{q+1}} \, b \, u_2$ 
with possibly other letters between $s_{r+2}$ and $s_q$ 
(if there is only one $ab$, apply the argumentation while wrapping around $w$ circularly, i.e., $s_r[w\rangle s_q$). 
Since $p$ solves ESSPs $\lnot s_r [b\rangle$ and $\lnot s_q [b \rangle$, the next inequalities hold true, 
where $\mu_r$ denotes the marking of $p_{a,b}$ at state $s_r$: 
\begin{equation}\label{system.eq}
\begin{aligned}
	\lnot s_r [b \rangle: & \hspace*{4mm} \mu_r 												& <\; & n \\
	s_{r+1} [b \rangle:  & \hspace*{4mm} \mu_r + m 									& \ge\; & n  \\
	\forall j:\; r \leq j \leq q: & \hspace*{4mm}  \mu_r + \Parikh_{rj}(a)\cdot m - \Parikh_{rj}	(b)\cdot n  &\ge\; & 0 \\
	\lnot s_q [b \rangle: & \hspace*{4mm} \mu_r + \Parikh_{rq}(a)\cdot m - \Parikh_{rq}(b)\cdot n & <\; & n\\
\end{aligned}
\end{equation} 

From the first and the third line of (\ref{system.eq}) we get $\Parikh_{rj}(a)\cdot m- \Parikh_{rj}(b)\cdot n > - n$. 
This implies:
\begin{equation}\label{interval1.eq}
	 \frac{\Parikh_{rj}(b)-1}{\Parikh_{rj}(a)} < \frac{m}{n}\mbox{ when } r< j \le q. 
\end{equation}
From the third and the fourth line of (\ref{system.eq}) we obtain 
 \[(\Parikh_{rq}(a) - \Parikh_{rj}(a))\cdot m - (\Parikh_{rq}(b) -  \Parikh_{rj}(b))\cdot n < n.
 \]
If $\Parikh_{jq}(a) \neq 0$, since $\Parikh_{rq}=\Parikh_{rj}+\Parikh_{jq}$ this inequality can be written as 
\begin{equation}\label{interval2.eq}
	 \frac{m}{n} < \frac{\Parikh_{jq}(b)+1}{\Parikh_{jq}(a)}.
\end{equation}
Thus, from~(\ref{interval1.eq})
and~(\ref{interval2.eq}) we have a necessary condition for solvability. 
\end{proof}

\vspace*{\baselineskip}
In particular, Lemma~\ref{nec.lem} explains the cyclic unsolvability of the word
$w_3 = \, \cstate{s_r} \,ab \, c \, b \, \cstate{s_j} \,   a \, d \, \cstate{s_q} \, ab \, d$.
Indeed, $\Parikh(w_3)(b)=3=\Parikh(w_3)(a)$, so that $m/n=1$
and $1 \nless \frac{0+1}{1} = \frac{\Parikh_{jq}(b)+1}{\Parikh_{jq}(a)}.$ 
Moreover, the necessary condition for cyclic sovability from Lemma~\ref{nec.lem} extends to a sufficient one in the following sense. 

\begin{lemma}[A sufficient condition for cyclic solvability by a WMG]\label{suf.lem} 
If $w \in T^\ast$ has a prime Parikh vector, and for each circularly adjacent $ab$ pair in 
	$ w = \ldots \, \cstate{s_q} \, a \, b \, \ldots, $ 
the inequality  
\begin{equation}\label{halfinterval.eq} 
	 \frac{m}{n} < \frac{\Parikh_{jq}(b)+1}{\Parikh_{jq}(a)} 
\end{equation}
holds true for any $s_j$ such that $\Parikh_{jq}(a)\neq 0$, then $w$ is cyclically WMG-solvable. 
\end{lemma}
\begin{proof}
We have proved in Lemma~\ref{SSP.lem} that all SSPs are solvable for prime cycles. 
Let us consider the ESSPs at states $s$ as in $w = \ldots \cstate{s} a b \ldots$, i.e. $\lnot s[b\rangle$. 
Since we are looking for a WMG solution, all the sought places are of the form $p_{a,b}$ (see Lemma~\ref{adjacent.lem} 
and Fig.~\ref{genplaceab.fig}) with $m,n$ as in~(\ref{mn.eq}). 
To define the initial marking of $p_{a,b}$, let us put $n\cdot \Parikh(w) (b)$ tokens on it and fire the sequence $w$ once completely. 
Choose some state $s'$ in $w = \ldots \cstate{s'} \, a  \ldots$ 
such that the number $k$ of tokens on $p_{a,b}$ at state $s'$ is minimal 
(it may be the case that such an $s'$ is not unique; we can choose any such state).
Define $M_0(p_{a,b}) = n\cdot \Parikh(w) (b)-k$ as the initial marking of $p_{a,b}$.  
By construction, the firing of $w$ reproduces the markings of $p_{a,b}$ and $M_0$ guarantees their non-negativity. 
Let us show that the constructed place $p_{a,b}$ solves all the ESSPs $\lnot s[b\rangle$, where $w = \ldots \cstate{s} a b \ldots$.
Consider such a state $s$ in $w$ (w.l.o.g. we assume $s \neq s'$, since $s'$ certainly disables $b$). 
From $ w = u_1 \, \cstate{s'} a  \, \ldots \, \cstate{s} \, a  \, b \, u_2$ (circularly) 
and from inequality~(\ref{halfinterval.eq}) for $s_j = s'$ and $s_q=s$, 
we get $\Parikh_{s' s}(a)\cdot m - \Parikh_{s' s}(b) \cdot n < n$ since $\Parikh_{s' s}(a)>0$. 
Since $ M_{s'}(p_{a,b}) = 0$, 
$M_s(p_{a,b}) = M_{s'}(p_{a,b}) + \Parikh_{s' s}(a)\cdot m - \Parikh_{s' s}(b) \cdot n < n$, i.e., 
$p_{a,b}$ disables $b$ at state~$s$.

Now, we show that places of the form $p_{a,b}$ also solve the other ESSPs against $b$, 
i.e., at the states where $b$ is not the subsequent transition.
Sequence $w$ (up to rotation) 
can be written as $w = u_1\, b^{x_1} \, u_2 \, b^{x_2}\, \ldots \, u_l \, b^{x_l}$, $1 \le l \le \Parikh(w) (b)$, 
and for $1 \le i \le l$: $x_i >0$, $u_i \in (T\setminus \{b\})^+$. 
Transition $b$ has to be disabled at all the states between successive $b$-blocks. 
Consider an arbitrary pair of such blocks $b^{x_j}$ and $b^{x_{j+1}}$ in 
$w   \; =\;   \ldots\, b^{x_j}\, u_j \, b^{x_{j+1}} \, \ldots \; = \;   \ldots\, b^{x_j}\, \cstate{s}\, u'_j\,  \cstate{r} \, t \,  b^{x_{j+1}} \, \ldots $, 
with $u_j = u'_j t$. Place $p_{t,b}$ does not allow $b$ to fire at state $r$.
We have to check that $b$ is not enabled at any state between $s$ and $r$, i.e., it is not enabled `inside' $u'_j$.
If $u'_j$ is empty, then $s=r$, and we are done. Let $u'_j \neq \varepsilon$.
Due to $\Parikh(u'_j)(b) = \Parikh(u_j)(b) = 0$, the marking of place $p_{t,b}$ cannot decrease from $s$ to $r$, 
i.e., 
$M_{s}(p_{t,b}) \le M_{s''}(p_{t,b}) \le M_r(p_{t,b})$ for any $s''$ `inside' $u'_j$.
Since $p_{t,b}$ disables $b$ at $r$, it then disables $b$ at all states between $s$ and $r$, inclusively.
\end{proof}

\vspace*{\baselineskip}
From Lemma~\ref{nec.lem} and Lemma~\ref{suf.lem} we can deduce the following characterisation.

\begin{theorem}[A characterisation of cyclic WMG-solvability]\label{wmg-char.th}
A sequence $w \in T^\ast$ is cyclically solvable with a WMG iff 
$\Parikh(w) $ is prime and, for any pair of circularly adjacent labels in $w$, 
for instance $w = \ldots \, \cstate{s_q}\, ab\, \ldots$,
\begin{equation*}
	 \frac{m}{n} < \frac{\Parikh_{jq}(b)+1}{\Parikh_{jq}(a)} 
\end{equation*}
holds true with $m$, $n$ as in~(\ref{mn.eq}) for any $s_j$ such that $\Parikh_{jq}(a) \neq 0$. 
A WMG solution can be found with the places of the form $p_{a,b}$ for every such pair of $a$ and $b$. 
\end{theorem}

\subsection{A Polynomial-time Algorithm for Cyclic WMG-Solvability}\label{firstalgo.subsec} 

From the characterisation given by Theorem~\ref{wmg-char.th} and the considerations above, 
Algorithm~\ref{wmg.alg} below synthesises a cyclic WMG solution for a given sequence $w \in T^*$, if one exists.

The algorithm works as follows. 
Initially, the Parikh vector of the input sequence is calculated and checked for primeness in lines 2-3. 
If the Parikh vector is prime, we consecutively consider all the pairs of adjacent letters 
and examine the inequality from Theorem~\ref{wmg-char.th} for them. 
To achieve it, we take the first two letters in the sequence (lines 4-11), 
check if the inequality is satisfied for all the states (lines 12-18), 
and if so, construct a new place connecting the two letters under consideration (19-25).
Then, the sequence is cyclically rotated such that the initial letter goes to the end 
and the second letter becomes initial (line 8). 
In the new sequence, we take again the first two letters (lines 9-11) and repeat the procedure for them. 
The algorithm stops after a complete rotation of the initial sequence, and by this moment all the pairs of adjacent letters have been checked. 
The ordered alphabet is stored in the array $T$, and the sequence is stored in $v$.
We use variables $a$ and $b$ to store the letters under consideration in each step, $ia$ and $ib$ to store their indices in the alphabet, 
and $na$ and $nb$ are used for counting their occurrences during the check of the inequality from Theorem~\ref{wmg-char.th}. 
Variables $M$ and $Mmin$ are used to compute the initial marking of the sought place.\\

\begin{algorithm}[h] 
\caption{Synthesis of a WMG solving a cyclic word} 
\label{wmg.alg}
 \SetKwInOut{Input}{input}\SetKwInOut{Output}{output}
  
  \Input{$w \in T^n,\; T = \{t_0, \ldots , t_{m-1}\}$} 
  \Output{A WMG system $(N,M_0)$ cyclically solving $w$, if it exists}  
  var: $T[0 \; .. \, m-1] = (t_0, \ldots , t_{m-1})$, 
  $v[0 \, .. \, n -1]$, 
  $a$, $b$, $na$, $nb$, $ia$, $ib$, $M$, $Mmin$\;
  compute the Parikh vector $\Parikh[0\, . .\, m-1]$ of $w$\;
  \lIf(\tcp*[f]{Parikh-primeness}){ $\Parikh$ is not prime}{
  	\KwRet{unsolvable}  
  }
  $b \leftarrow w[0]$\; 
  \For(\tcp*[f]{index of $b$} ){$j = 0$ \KwTo $m -1$}{  
  	\lIf{$b = T[j]$}{
		$ib \leftarrow j$ 
	}  
  }
  \For{$i =0$ \KwTo $ n -1$ }{ 
  	$ v \leftarrow w[i] \ldots w[n -1] w[0] \ldots w[i-1]$
	\tcp*{rotation of $w$}
	$a \leftarrow b$, $b\leftarrow v[1]$, $ia \leftarrow ib$
	\tcp*{fix first adjacent pair}
	\For{$j=0$ \KwTo $m -1$}{
		\lIf{$b = T[j]$}{
			$ib \leftarrow j$
		}
	}
	$na \leftarrow 1$, $nb \leftarrow 1$ \;
	\If{$a \neq b$}{
		\For{$k=2$ \KwTo $ n -1$}{
			\If{$\frac{\Parikh[ib]}{\Parikh[ia]} \geq \frac{\Parikh[ib]-nb+1}{\Parikh[ia]-na}$ }{
				\KwRet{unsolvable} \tcp*{check solvability condition}
			}
			\lIf{$v[k] = T[ia]$}{
				$na \leftarrow na +1$
			}
			\lIf{$v[k] = T[ib]$}{
				$nb \leftarrow nb +1$
			}
		}
		$M\leftarrow \Parikh[ia] \cdot \Parikh[ib]$, $Mmin \leftarrow M$\;
		\For(\tcp*[f]{find initial marking}){$k =0$ \KwTo $ n -1$}{
			\lIf{$w[k]=a$}{
				$M\leftarrow M+ \Parikh[ib]$
			}
			\lIf{$w[k]=b$}{
				$M\leftarrow M- \Parikh[ia]$
			}
			\lIf{$M < Mmin$}{
				$Mmin \leftarrow M$
			}
		}
		add new place $p_{T[ia],T[ib]}$ to $N$ with\\
		$W(T[ia],p) = \Parikh[ib]$, $W(p,T[ib]) = \Parikh[ia]$, $M_0 = \Parikh[ia] \cdot \Parikh[ib] - Mmin$\;
	}
  }
  \KwRet{ $(N,M_0)$} 
\end{algorithm}

\noindent {\bf Polynomial-time complexity of Algorithm~\ref{wmg.alg}}. 
For a sequence of length $n$ over an alphabet with $m$ labels,
the Parikh vector can be computed in $\mathcal{O}(n)$ 
and its primeness can be checked using e.g. the Euclidean algorithm, with a running time in $\mathcal{O}(m\cdot \log_2^2 n)$. 
The main $\mathtt{for}$-cycle of Algorithm~\ref{wmg.alg} involves the enumeration of all pairs of distinct states of the cycle. 
For each pair of adjacent labels, a run of the $\mathtt{for}$-cycle consists of a lookup for an index in $\mathcal{O}(m)$, 
the verification of the inequality in $\mathcal{O}(n)$ and the construction of a place in $\mathcal{O}(n)$, which sums up to $\mathcal{O}(m+n)$. 
Thus, the main $\mathtt{for}$-loop requires a runtime in $\mathcal{O}(n(n+m))$.
Taking into account that $m \le n$, and that $n$ growths much faster than $\log_2^2 n$,  
the overall running time of the algorithm does not exceed $\mathcal{O}(n^2)$.\\

\noindent {\bf Complexity comparison: the known general approach is less efficient.} 
For a comparison, 
solving a cycle of length $n$ over $m$ labels with a WMG amounts to solve $n(n-1)$ SSPs and $n(m-1)$ ESSPs at most. 
In the special case of WMG synthesis from a prime cycle, we know that all the SSPs are solvable (Lemma~\ref{SSP.lem})
and that solving first the other problems avoids to consider the SSPs (see~\cite{DH2018}).
Since each of the sought places has at most one input and one output, 
each of the separation problems seeks for $3$ unknown variables, namely the initial marking of a place, the input and the output arc weights. 
For an ESSP, the output transition is clearly the one which has to be disabled and the input transition is to be found. 
So, there are $m-1$ possibilities to define a concrete ESSP, 
which in the worst case gives us up to $n(m-1)^2$ systems of inequalities to solve all the ESSPs.

The general region-based synthesis typically uses ILP-solvers,
and using e.g. Karmarkar's algorithm~\cite{Karmarkar84}  
(which is known to be efficient) for solving  an ILP-problem with $k$ unknowns, we expect a running time of 
$\mathcal{O}(k^{3.5}\cdot L^2 \cdot \log L \cdot \log \log L)$ where $L$ is the length of the input in bits. 
For the case of a cycle, the input of each separation problem is 
the matrix with the range of $(n+1) \times m$ and the vector of right sides with the range of $n+1$, 
where each component of the vector and of the matrix is a natural number not greater than $n$. 
Hence, the length of the input for a single separation problem can be estimated as $L = (m +1)\cdot (n+1)\cdot  \log_2 n$ bits,  
implying a runtime of 
$\mathcal{O}(n^2\cdot m^2 \cdot \log^2_2 n \cdot \log  L \cdot \log \log L)$ 
for solving a single separation problem (the number of unknowns being equal to 3, i.e. constant). 
Thus the general synthesis approach would need a runtime of $\mathcal{O}(n^3 \cdot m^4 \cdot LF)$ 
with the logarithmic factor 
$LF =\log_2^2 n \cdot \log ((m+1)\cdot (n+1)\cdot  \log_2 n) \cdot \log \log ((m+1)\cdot (n+1)\cdot  \log_2 n)$. 
Note that, with this general approach, some redundant places may be constructed, 
but they can be wiped out in a post-processing phase.

\subsection{CF-solvability vs WMG-solvability of Cycles}\label{nocf.subsec} 

Let us now relate cyclic WMG-solvability to cyclic CF-solvability. 

\begin{theorem}\label{cf=wmg.th}
A sequence $u \in \{ a,b,c\}^\ast $ is cyclically WMG-solvable iff $u$ is cyclically CF-solvable.
\end{theorem}
\begin{proof}
WMGs form a proper subclass of CF nets, hence the direct implication. 
Let now $TS = (S, T = \{a,b,c\}, \to, s_0)$ be a CF-solvable circular \lts{} obtained from $u$ and $\uniqueP = \Parikh(u)$. 
By contraposition, assume that $TS$ is not solvable by a WMG. 
Then, due to Theorem~\ref{wmg-char.th}, for some distinct states $j, q \in S$ and distinct labels $a,b \in T$
\begin{equation}\label{cf=wmg.eq} 
	\frac{\uniqueP(a)}{\uniqueP(b)} \ge \frac{\Parikh_{jq}(a)+1}{\Parikh_{jq}(b)}. 
\end{equation}
W.l.o.g. we can choose the leftmost $j$ satisfying~(\ref{cf=wmg.eq}). 
Then, in $TS$ we have $r[a\rangle j$ for some $r\in S$ preceding $j$. 
Indeed, if this is not the case and either $r[b\rangle j$ or $r[c \rangle j$, then~(\ref{cf=wmg.eq}) holds true for $r$ and $q$,  
contradicting the choice of $j$. 
On the other hand, since $\Parikh_{jq}(a) +1 = \Parikh_{rq}(a)$ and $\Parikh_{jq}(b) = \Parikh_{rq}(b)$, 
the inequality~(\ref{cf=wmg.eq}) implies 
\begin{equation}\label{cf=wmg1.eq}
	\frac{\uniqueP(a)}{\uniqueP(b)} \ge \frac{\Parikh_{rq}(a)}{\Parikh_{rq}(b)}.
\end{equation}

Consider a place $p$ which is an input place of $a$ in a cyclic CF solution of $u$.
From Lemma~\ref{pure-CF.lem}, we can assume pureness, i.e., 
the place has the form illustrated on the right of Fig.~\ref{cf=wmg.fig} with $x = a$, $y=b$, $z=c$. 
\begin{figure}[h] 
\centering
\begin{tikzpicture}[scale=0.69]
\begin{scope}[xshift=0cm,yshift=0cm]
\node[circle,fill=black!100,inner sep=0.07cm](01)at(180:1.45)[]{};
\node[circle,fill=black!100,inner sep=0.05cm](02)at(120:1.45)[label=below right:$q$]{};
\node[circle,fill=black!100,inner sep=0.05cm](03)at(60:1.45)[]{};
\node[circle,fill=black!100,inner sep=0.05cm](04)at(0:1.45)[]{};
\node[circle,fill=black!100,inner sep=0.05cm](05)at(300:1.45)[label=above left:$r$]{};
\node[circle,fill=black!100,inner sep=0.05cm](06)at(240:1.45)[label=above right:$j$]{};
\draw[-latex](-2.,0)to node[]{}(01);
\draw[dotted,-latex](01)to node[auto,above left]{$ $}(02);
\draw[-latex](02)to node[auto,above]{$b$}(03);
\draw[-latex](03)to node[auto,above right]{$a$}(04);
\draw[dotted,-latex](04)to node[auto,right]{$ $}(05);
\draw[-latex](05)to node[auto,below,pos=0.4]{$a$}(06);
\draw[dotted,-latex](06)to node[auto,below]{$ $}(01);
\end{scope}
\begin{scope}[xshift=4cm,yshift=0cm]
\node[place](p)at(2,0){\scriptsize $\mu_0$};
\node[below of=p,node distance=20pt]{$p$}; 
\node[transition] (a) at (4,0) {$x$}; 
\node[transition] (b) at (0,1) {$y$};
\node[transition] (c) at (0,-1) {$z$};
\draw(p)[]edge[-latex,out=0,in=180]node[below,swap,inner sep=4pt,pos=0.45]{$k_x$}(a);
\draw(b)[]edge[-latex]node[auto,inner sep=1pt,pos=0.45]{$k_y$}(p);
\draw(c)[]edge[-latex]node[below,inner sep=4pt,pos=0.65]{$k_z$}(p);
\end{scope}
\end{tikzpicture}
\vspace*{-1mm}
\caption{$u$ (left) is cyclically solvable with a CF net; a CF place over $\{x,y,z\}$ (right).}
\label{cf=wmg.fig}
\end{figure}
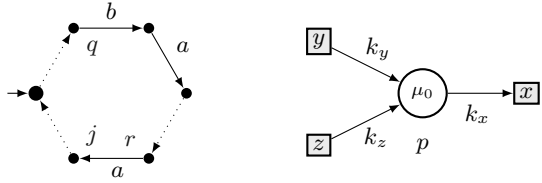
We must have the following constraints for $p$: 
\begin{equation}\label{cyclicalityA.eq}
\begin{aligned}
	 & \text{cycle }&: &\;\; k_b \cdot \uniqueP(b) + k_c \cdot \uniqueP(c) = k_a \cdot \uniqueP(a) \\
	 & r[a\rangle &: &\;\;  M_r(p) \ge k_a \\
	 & r[a \ldots\rangle q &: &\;\; M_q(p) = M_r(p) + k_b \cdot \Parikh_{rq}(b) + k_c \cdot \Parikh_{rq}(c) - k_a \cdot \Parikh_{rq}(a) .
\end{aligned}
\end{equation}

If $\Parikh_{rq}(c) \ge \Parikh_{rq}(a) \cdot \frac{\uniqueP(c)}{\uniqueP(a)}$, then due to (\ref{cf=wmg1.eq}) and (\ref{cyclicalityA.eq}),
\begin{equation*}
\begin{aligned} 
	M_q(p) & = & M_r(p) + k_b \cdot \Parikh_{rq}(b) + k_c \cdot \Parikh_{rq}(c) - k_a \cdot \Parikh_{rq}(a) \\ 
	            & \ge & k_a + \big( k_b \cdot \frac{\uniqueP(b)}{\uniqueP(a)} + k_c \cdot \frac{\uniqueP(c)}{\uniqueP(a)} - k_a \big)\cdot \Parikh_{rq}(a)  = k_a ,
\end{aligned}
\end{equation*}
implying $q [ a\rangle$ which contradicts $q[ b\rangle$. 
Hence, $\Parikh_{rq}(c) < \Parikh_{rq}(a) \cdot \frac{\uniqueP(c)}{\uniqueP(a)}$. 
Together with~(\ref{cf=wmg1.eq}), we have 
\[  
	\frac{\Parikh_{rq}(b)}{\uniqueP(b)} \ge \frac{\Parikh_{rq}(a)}{\uniqueP(a)}> \frac{\Parikh_{rq}(c)}{\uniqueP(c)}. 
\]
which is equivalent to 
\begin{equation}\label{dualcyc.eq}
	\frac{\Parikh_{qr}(b)}{\uniqueP(b)} \le \frac{\Parikh_{qr}(a)}{\uniqueP(a)}< \frac{\Parikh_{qr}(c)}{\uniqueP(c)} .
\end{equation}
For an arbitrary input place of $b$, hence of the form illustrated on the right of Fig.~\ref{cf=wmg.fig} with $x = b$, $y=a$, $z=c$, 
\begin{equation}\label{cyclicalityB.eq}
\begin{aligned}
	 & \text{cycle }&: &\;\; k_a \cdot \uniqueP(a) + k_c \cdot \uniqueP(c) = k_b \cdot \uniqueP(b) \\
	 & q[b\rangle &: &\;\;  M_q(p) \ge k_b \\
	 & q[b \ldots\rangle r &: &\;\; M_r(p) = M_q(p) + k_a \cdot \Parikh_{qr}(a) + k_c \cdot \Parikh_{qr}(c) - k_b \cdot \Parikh_{qr}(b) .
\end{aligned}
\end{equation}
Then, due to (\ref{dualcyc.eq})  and (\ref{cyclicalityB.eq}),
\begin{equation*}
\begin{aligned} 
	M_r(p) & = & M_q(p) + k_a \cdot \Parikh_{qr}(a) + k_c \cdot \Parikh_{qr}(c) - k_b \cdot \Parikh_{qr}(b) \\  
                    & \ge & k_b + \big( k_a \cdot \frac{\uniqueP(a)}{\uniqueP(b)} + k_c \cdot \frac{\uniqueP(c)}{\uniqueP(b)} - k_b \big) \cdot \Parikh_{qr}(b) = k_b
\end{aligned}
\end{equation*}
which implies $r [b\rangle$, a contradiction. Thus, $TS$ is solvable by some WMG.
\end{proof}

\vspace*{\baselineskip}
When the alphabet has more than three elements,
the inclusion of WMGs into CF nets is strict, 
i.e., there are cyclically CF-solvable sequences that are not cyclically WMG-solvable: 
the sequence $w = abcbad$ has a cyclic CF solution (cf. Fig.~\ref{cf-no-wmg.fig}); 
for $a\, \cstate{r}\, b\, c\, \cstate{q}\, b\, a\, d$ we have 
$ \frac{\Parikh(w)(a)}{\Parikh(w)(b)} = \frac{2}{2} \nless \frac{0+1}{1} = \frac{\Parikh_{rq}(a)+1}{\Parikh_{rq}(b)}$
which, by Theorem~\ref{wmg-char.th}, implies the cyclic unsolvability of $w$ by a WMG.

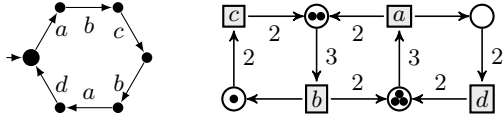
\begin{figure}[!ht]
\centering
\begin{tikzpicture}[scale=0.55]
\begin{scope}[xshift=0cm,yshift=0cm]
\node[circle,fill=black!100,inner sep=0.08cm](00)at(180:1.4)[]{};
\node[circle,fill=black!100,inner sep=0.05cm](01)at(120:1.4)[]{};
\node[circle,fill=black!100,inner sep=0.05cm](02)at(60:1.4)[]{};
\node[circle,fill=black!100,inner sep=0.05cm](03)at(0:1.4)[]{};
\node[circle,fill=black!100,inner sep=0.05cm](04)at(300:1.4)[]{};
\node[circle,fill=black!100,inner sep=0.05cm](05)at(240:1.4)[]{};

\draw[-latex](-2.,0)to node[]{}(00);
\draw[-latex](00)to node[auto,right]{$a$}(01);
\draw[-latex](01)to node[auto,below]{$b$}(02);
\draw[-latex](02)to node[auto,left]{$c$}(03);
\draw[-latex](03)to node[auto,left]{$b$}(04);
\draw[-latex](04)to node[auto,above]{$a$}(05);
\draw[-latex](05)to node[auto,right]{$d$}(00);
\end{scope}
\begin{scope}[xshift=3.5cm,yshift=-1.cm]
\node[place,tokens=1](p0)at(0.,0.)[]{$ $};
\node[place,tokens=2](p1)at(2.,2.)[]{$ $};
\node[place,tokens=3](p2)at(4.,0.)[]{$ $};
\node[place,tokens=0](p3)at(6.,2.)[]{$ $};
\node[transition](a)at(4.,2.){$a$}
	edge[post]node[below]{$2$}(p1)
	edge[post]node[above]{$ $}(p3)
	edge[pre]node[right]{$3$}(p2)
	;
\node[transition](b)at(2.,0.){$b$}
	edge[pre]node[right]{$3$}(p1)
	edge[post]node[above]{$2$}(p2)
	edge[post]node[below right]{$ $}(p0)
	;
\node[transition](c)at(0,2.){$c$}
	edge[post]node[below]{$2$}(p1)
	edge[pre]node[right]{$2$}(p0)
	;
\node[transition](d)at(6.,0.){$d$}
	edge[pre]node[right]{$2$}(p3)
	edge[post]node[above]{$2$}(p2)
	;
\end{scope}
\end{tikzpicture}
\caption{Sequence $abcbad$ is cyclically solved by the CF net on the right.}
\label{cf-no-wmg.fig}
\end{figure}

By Lemma~\ref{adjacent.lem}, using places only between adjacent transitions is sufficient for cyclic WMG-solvability. 
For the sequence $abcbad$ in Fig.~\ref{cf-no-wmg.fig}, $b$ follows $a$ and $c$, 
and the input place of $b$ in the CF solution is an output place for both $a$ and $c$. 
The situation is similar for $a$, which follows $b$ and $d$. 
However, this is not always the case when we are looking for a solution in the class of CF nets. 
For instance, the sequence $cabdaaab$ is cyclically solvable by a CF net (see Fig.~\ref{cf-no-adj.fig}). 
In this sequence, $b$ always follows $a$. 
But in order to solve ESSPs against $b$, we need an output place for $c$ (in addition to~$a$). 
Indeed, if there is a place $p_{a,b}$ as on the  right of Fig.~\ref{cf-no-adj.fig} 
which solves ESSPs against $b$, then 
for $c a\, \cstate{s} \,bdaa\, \cstate{q} \,ab$  we get 
\begin{equation*}
\begin{aligned}
	&s[b\rangle &: \,\,	&\mu_0 + 2 & \ge & ~4\\
	&\lnot q[b\rangle &: \,\, &\mu_0 + 3\cdot 2 - 4 &<& ~4\\
\end{aligned}
\end{equation*}
Subtracting the first inequality from the second one, we get $4-4 < 0$, a contradiction. 
Hence, $p_{a,b}$ cannot solve all ESSPs against $b$ in the cycle $cabdaaab$.

\begin{figure}[h]  
\centering
\begin{tikzpicture}[scale=0.55]
\begin{scope}[xshift=0cm,yshift=0cm]
\node[circle,fill=black!100,inner sep=0.08cm](00)at(180:1.45)[]{};
\node[circle,fill=black!100,inner sep=0.05cm](01)at(135:1.45)[]{};
\node[circle,fill=black!100,inner sep=0.05cm](02)at(90:1.45)[]{};
\node[circle,fill=black!100,inner sep=0.05cm](03)at(45:1.45)[]{};
\node[circle,fill=black!100,inner sep=0.05cm](04)at(0:1.45)[]{};
\node[circle,fill=black!100,inner sep=0.05cm](05)at(315:1.45)[]{};
\node[circle,fill=black!100,inner sep=0.05cm](06)at(270:1.45)[]{};
\node[circle,fill=black!100,inner sep=0.05cm](07)at(225:1.45)[]{};

\draw[-latex](-2.3,0)to node[]{}(00);
\draw[-latex](00)to node[auto,left]{$c$}(01);
\draw[-latex](01)to node[auto,above left]{$a$}(02);
\draw[-latex](02)to node[auto,above]{$b$}(03);
\draw[-latex](03)to node[auto,above right]{$d$}(04);
\draw[-latex](04)to node[auto,right]{$a$}(05);
\draw[-latex](05)to node[auto,below]{$a$}(06);
\draw[-latex](06)to node[auto,below]{$a$}(07);
\draw[-latex](07)to node[auto,left]{$b$}(00);
\end{scope}
\begin{scope}[xshift=3.5cm,yshift=-1.5cm]
\node[place,tokens=2](p0)at(1.5,1.5)[]{$ $};
\node[place,tokens=1](p1)at(0,2.8)[]{$ $};
\node[place,tokens=0](p3)at(4.5,2.8)[]{$ $};
\node[place,tokens=0](p4)at(4.5,0.2)[]{$ $};
\node[transition](a)at(4.5,1.5){$a$}
	edge[post]node[above]{$ $}(p4)
	edge[pre]node[above]{$ $}(p3)
	;
\node[transition](b)at(0.,0.2){$b$}
	edge[post]node[above left,inner sep=1pt]{$ $}(p1)
	edge[post]node[above left]{$ $}(p0)
	edge[pre]node[above,inner sep=2pt]{$3$}(p4)
	;
\node[transition](c)at(3,1.5){$c$}
	edge[post]node[above]{$ $}(p3)
	edge[pre]node[above]{$2$}(p0)
	edge[post]node[left]{$2$}(p4)
	;
\node[transition](d)at(2,2.8){$d$}
	edge[pre]node[below]{$2$}(p1)
	edge[post]node[below, near start]{$3$}(p3)
	;
\end{scope}
\begin{scope}[xshift=10cm,yshift=0cm]
\node[place,inner sep=1pt](p)at(2,0)[label=below:$p_{a,b}$]{\scriptsize $\mu_0$}; 
\node[transition](a)at(0,0){$a$};
\node[transition](b)at(4,0){$b$};
\draw(a)[]edge[-latex]node[above]{$2$}(p);
\draw(p)[]edge[-latex]node[above]{$4$}(b);
\end{scope}
\end{tikzpicture}
\vspace*{-1mm}
\caption{$cabdaaab$ is cyclically CF-solvable (middle), but is not cyclically WMG-solvable.} 
\label{cf-no-adj.fig}
\end{figure}
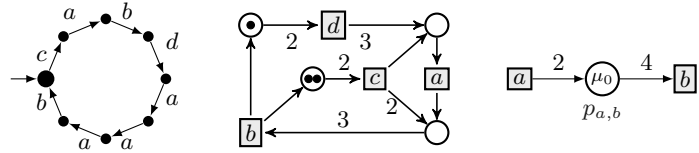

In a WMG, a place has at most one input. This restriction is relaxed for CF nets: multiple inputs are allowed. 
Let us show that a single input place for each transition is not always sufficient. 
For instance, consider the cyclically CF-solvable sequence $bcaf deaa abcd aafd caaa$ and Fig.~\ref{cf-big.fig}. 
Assume we can solve all ESSPs against transition $a$ with a single place $p$ as on the right of the same figure;  
due to Lemma~\ref{pure-CF.lem}, we do not need any side-condition. 
\begin{figure}[!ht] 
\centering
\begin{tikzpicture}[scale=0.6]
\begin{scope}[xshift=0cm,yshift=-2.8cm]
\node[place,tokens=0,inner sep=1pt](p0)at(-0.5,6.)[]{\scriptsize $7$};
\node[place,tokens=1](p1)at(9.,0.)[]{$ $};
\node[place,tokens=1](p2)at(9.,6.)[]{$ $};
\node[place,tokens=0](p3)at(3.2,3.)[]{$ $};
\node[place,tokens=0](p4)at(5,4.5)[]{$ $};
\node[place,tokens=2](p5)at(8,3)[]{$ $};
\node[place,tokens=2](p6)at(1.5,1.5)[]{$ $};
\node[place,tokens=3](p7)at(5,0.75)[]{$ $};
\node[place,tokens=1](p8)at(6.5,5.25)[]{$ $};
\node[place,tokens=0,inner sep=1pt](p9)at(1.5,4.5)[]{\scriptsize $10$};
\node[transition](a)at(1.5,6.){$a$}
	edge[post]node[below]{$2$}(p0)
	edge[post]node[left]{$2$}(p9)
	edge[pre]node[right]{$ $}(p2)
	edge[pre]node[right]{$ $}(p8)
	edge[pre]node[below left]{$4$}(p4)
	;
\node[transition](b)at(1.5,3.){$b$}
	edge[post]node[above,inner sep=2pt]{$3$}(p3)
	edge[post]node[left]{$3$}(p6)
	edge[pre]node[left,inner sep=2pt]{$9$}(p9)
	;
\node[transition](c)at(5,3){$c$}
	edge[pre]node[above]{$2$}(p3)
	edge[pre]node[below,near start]{$5$}(p6)
	edge[post]node[left]{$9$}(p4)
	edge[post]node[right,pos=0.4]{$3$}(p7)
	;
\node[transition](d)at(9.,1.5){$d$}
	edge[pre]node[left]{$2$}(p1)
	edge[pre]node[below]{$5$}(p7)
	edge[post]node[above,pos=0.35]{$3$}(p6)
	edge[post]node[above]{$ $}(p5)
	edge[post]node[left]{$3$}(p2)
	;
\node[transition](e)at(6.5,3.){$e$}
	edge[pre]node[above]{$3$}(p5)
	edge[post]node[above right,inner sep=1pt,pos=0.4]{$9$}(p4)
	edge[post]node[right]{$9$}(p8)
	;
\node[transition](f)at(-0.5,0.){$f$}
	edge[pre]node[right]{$9$}(p0)
	edge[post]node[above,inner sep=2pt,pos=0.5]{$3$}(p7)
	edge[post]node[below,inner sep=2pt]{$3$}(p1)
	;
\end{scope}
\begin{scope}[xshift=12cm,yshift=0cm]
\node[place,tokens=0](p)at(2,0)[label=below:$p$]{$\mu_0$};
\node[transition](b)at(0,2.5){$b$};
\node[transition](c)at(-0.7,1.5){$c$};
\node[transition](d)at(-1.3,0){$d$};
\node[transition](e)at(-0.7,-1.5){$e$};
\node[transition](f)at(0,-2.5){$f$};
\node[transition](a)at(4,0){$a$};
\draw(b)[]edge[-latex]node[right]{$k_b$}(p);
\draw(c)[]edge[-latex]node[above]{$k_c$}(p);
\draw(d)[]edge[-latex]node[above]{$k_d$}(p);
\draw(e)[]edge[-latex]node[above]{$k_e~~$}(p);
\draw(f)[]edge[-latex]node[ left]{$k_f$}(p);
\draw(p)[]edge[-latex,bend right=0]node[below]{$k$}(a);
\end{scope}
\end{tikzpicture}
\caption{$w=bcaf deaa abcd aafd caaa$ is cyclically solved by the CF net on the left;
a (pure) place of a CF net with $6$ transitions on the right.} 
\label{cf-big.fig}
\end{figure}
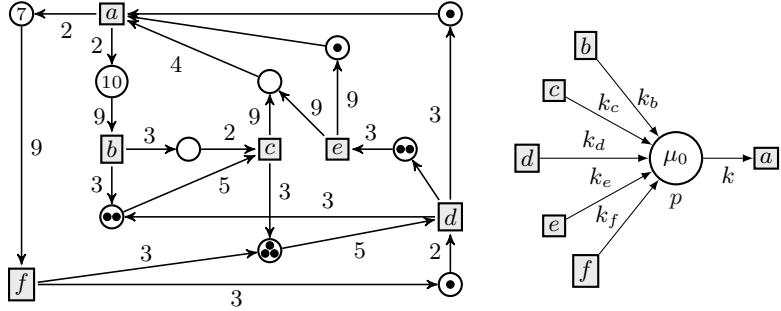
Then, for $p$ and 
$w=\cstate{s_0}\,b\, \,c\, \,a\, \,f\, \,d\,\cstate{s_5} \,e\,\cstate{s_6} \,a\, \,a\, \,a\, \,b\, \,c\,\cstate{s_{11}} \,d\,\cstate{s_{12}} \,a\, \,a\, \,f\, \,d\,\cstate{s_{16}} \, c\,\cstate{s_{17}} \,a\, \,a\, \,a$, 
the following system of inequalities must hold true: 
\begin{equation*}
\begin{aligned}
	& \text{cycle} &:\;\; & 2\cdot k_b + 3 \cdot k_c + 3 \cdot k_d + k_e + 2 \cdot k_f &=~& 9 \cdot k & (0) &\\
	& \lnot s_5[ a\rangle &: \;\; &  \mu_0 + k_b+ k_c+ k_d + k_f -k &<~& k & (1) &\\
	& s_6[aaa\rangle&: \;\; & \mu_0 + k_b+ k_c+ k_d + k_e + k_f -k &\ge~& 3\cdot k & (2) &\\
	& \lnot s_{11}[a\rangle &: \;\; & \mu_0 + 2\cdot k_b+ 2\cdot k_c+ k_d + k_e + k_f - 4\cdot k  &<~& k & (3) &\\
	& s_{12}[aa\rangle &: \;\; & \mu_0 + 2\cdot k_b+ 2\cdot k_c+ 2\cdot k_d + k_e + k_f - 4\cdot k  &\ge~& 2\cdot k& (4) &\\
	& \lnot s_{16}[a\rangle &: \;\; & \mu_0 + 2\cdot k_b+ 2\cdot k_c+ 3\cdot k_d + k_e + 2\cdot k_f - 6\cdot k  &<~& k & (5) &\\
	& s_{17}[aaa\rangle &: \;\; & \mu_0 + 2\cdot k_b+ 3\cdot k_c+ 3\cdot k_d + k_e + 2\cdot k_f - 6\cdot k &\ge~& 3\cdot k& (6) &\\
\end{aligned}
\end{equation*}

From the system above we obtain:
\begin{equation*}
\begin{aligned}
	& (2)-(1) & : \;\; & k_e & >~ & 2 \cdot k \\
	& (4)-(3) & : \;\; & k_d & >~ & k \\
	& (6)-(5) & : \;\; & k_c & >~ & 2 \cdot k \\
\end{aligned}
\end{equation*}
which implies $3\cdot k_c + 3\cdot k_d + k_e > 11 \cdot k$, contradicting the equality $(0)$. 
Hence, the ESSPs against $a$ cannot be solved by a single place.

\section{Weak Synthesis of WMGs in Polynomial-time}\label{WeakSynthesis}

For any given word $w$ over a set of labels $T$ whose support equals $T$,
each system $\mysyst=((P,T,W),M_0)$ that cyclically WMG-solves $w$, when it exists, 
has a unique minimal (hence prime) T-semiflow $\uniqueP$ with support $T$, since it is live 
(meaning that for each transition $t$, from each reachable marking $M$, a marking $M'$ is reachable from $M$ that enables $t$) 
and bounded (see~\cite{WTS92}).
In some situations, it might be sufficient to specify only the desired unique minimal T-semiflow, 
which leads to what we call a {\em weak synthesis} problem. 
Given such a prime Parikh vector $\uniqueP$, the aim is thus to construct a WMG cyclically solving 
an arbitrary sequence whose Parikh vector equals $\uniqueP$.
In this section, we provide a method for constructing a solution in polynomial-time.
To achieve it, we first need to recall known liveness conditions for WMGs and their circuit subclass.

\subsection{Previous results on liveness}

In~\cite{March09}, a polynomial-time sufficient condition of liveness is developed for the well-formed, 
strongly connected weighted event graphs (WEGs), equivalent to the well-formed, strongly connected WMGs.
Under these assumptions, each place has exactly one ingoing and one outgoing transitions.
Variants of this liveness condition for other classes of nets are given in~\cite{TECS14}, Theorems~4.2 and~5.5.\\

\noindent {\bf Additional notions.} We introduce the following notions for our purpose:\\
$-$ For any place $p$, $\gcd_p$ denotes the $\gcd$ of all input and output weights of $p$.\\
$-$ A marking $M_0$ satisfies the {\em useful tokens condition} if, for each place $p$, $M_0(p)$ is a multiple of $\gcd_p$.
Indeed, if $M_0(p) = k \cdot \gcd_p + r$ for some non-negative integers $k$ and $r$ such that $0 < r < \gcd_p$, 
then $r$ tokens are never used by any firing (see~\cite{March09,TECS14} for more details).\\
$-$ A net $(P,T,W)$ with incidence matrix $I$ is {\em conservative} 
if there is a P-vector $X \ge \oneP$ such that $X \cdot I = \zeroT$, 
where $\oneP$ denotes the vector of size $|P|$ in which each component has value $1$.
Such a P-vector $X$ is called a {\em conservativeness} vector.
The net is {\em $1$-conservative} if $\oneP$ is a conservativeness vector,
i.e. if for each transition, the sum of its input weights equals the sum of its output weights.\\
$-$ A net $N$ is {\em structurally bounded} if for each marking $M_0$, $(N,M_0)$ is bounded.\\
$-$ By Theorem 4.11 in~\cite{WTS92}, a live WTS $(N,M_0)$ is bounded iff $N$ is conservative.
We focus on live and bounded WMG solutions (which are WTS), hence on conservative, thus structurally bounded (see~\cite{LAT98}), solutions.\\
$-$ The scaling operation~(Definition 3.1 in \cite{TECS14}): 
The multiplication of all input and output weights of a place $p$ together with its marking by a positive rational number $\alpha_p$ 
is a {\em scaling} of the place $p$ if the resulting input and output weights and marking are integers. 
If each place $p$ of a system is scaled by a positive rational $\alpha_p$, the system is said to be scaled by the vector $\alpha$ 
whose components are the scaling factors $\alpha_p$.\\

Recalling Theorem~3.2 in~\cite{TECS14}, if $\mysyst = ((P, T, W), M_0)$ is a system 
and $\alpha$ is a vector of $|P|$ positive rational components, 
then scaling $\mysyst$ by $\alpha$ preserves the feasible sequences of firings.
We deduce from it, and from Theorem 3.5 in the same paper, the following result.

\begin{lemma}\label{SameEnabledPlaces}
Consider a system $\mysyst$, whose scaling by a vector $\alpha$ of positive rationals yields the system $\mysyst'$.
Then, for each feasible sequence $\sigma$, denote by $M$ the marking reached when firing $\sigma$ in $\mysyst$
and by $M'$ the marking reached when firing $\sigma$ in $\mysyst'$;
the set of places enabled by $M$ equals the set of places enabled by $M'$.
\end{lemma}

Thus, in conservative systems, we can reason equivalently on feasible sequences and enabled places
in the system scaled by a conservativeness vector, yielding a $1$-conservative system whose tokens amount remains constant.

Next result is a specialisation of Theorem 4.5 in~\cite{TECS14} to circuit Petri nets,
using the fact that the liveness of circuits is monotonic, i.e. preserved upon any addition of initial tokens 
(see \cite{WTS92,ACSD13,TECS14} and Theorem 7.10 in~\cite{HDFI18} with its proof).

\begin{proposition}[Sufficient condition of liveness~\cite{March09,TECS14,HDFI18}]\label{SuffCondCircuit}
Consider a conservative circuit system $\mysyst=(N,M_0)$, with $N=(P,T,W)$. 
$\mathcal{S}$ is live if the following conditions hold:\\
$-$ for a place $p_0$, with $\{t_0\} = p_0^\bullet$, $M_0(p_0) = W(p_0,t_0)$;\\
$-$ for every place $p$ in $P \setminus \{p_0\}$, with $p^\bullet = \{t\}$, $M_0(p) = W(p,t) - \gcd_p$.\\
Moreover, for every marking $M_0'$ such that $M_0' \ge M_0$, $(N,M_0')$ is live.
\end{proposition}

In the particular case of a binary circuit, i.e. with two transitions, 
we recall the next characterisation condition of liveness, given as Theorem~5.2 in~\cite{March09}.

\begin{proposition}[Liveness of binary $1$-conservative circuits~\cite{March09}]\label{LiveBinaryCircuit}
Consider a $1$-conservative binary circuit $\mysyst=((P,T,W),M_0)$ 
that fulfills the useful tokens condition,
with $T = \{a,b\}$ and
$P = \{p_{a,b},p_{b,a}\}$, where $p_{a,b}$ is the output of $a$ and $p_{b,a}$ is the output of $b$.
Let $m=W(a,p_{a,b})$ and $n=W(p_{a,b},b)$.
Then $\mathcal{S}$ is live iff $M_0(p_{a,b}) + M_0(p_{b,a}) > m + n - 2 \cdot \gcd(m,n)$. 
\end{proposition}

Now, consider any $1$-conservative binary circuit system $\mathcal{S}$ whose initial marking $M_0$ marks 
one place $p_{a,b}$ with its output weight $W(p_{a,b},b)$
and the other place $p_{b,a}$ with $W(p_{b,a},a) - \gcd_{p_{b,a}}$.
Each marking reachable from $M_0$ enables exactly one place;
applying Lemma~\ref{SameEnabledPlaces}, this is also the case for any scaling of $\mathcal{S}$, hence:

\begin{lemma}[One enabled place in binary circuits]\label{OneEnabledPlace}
Consider a conservative binary circuit system $\mysyst=(N,M_0)$, with $N=(P,T,W)$, $T=\{a,b\}$,
such that for the place $p_{a,b}$ with output $b$, $M_0(p_{a,b}) = W(p_{a,b},b)$ 
and 
for the other place $p_{b,a}$ with output $a$, $M_0(p_{b,a}) = W(p_{b,a},a) - \gcd_{p_{b,a}}$.
Then each reachable marking enables exactly one place.
\end{lemma}

This lemma will help ensuring that the reachability graph of the synthesised WMG forms a circle.
Fig.~\ref{solcyc2duplicate3.fig} illustrates Lemma~\ref{OneEnabledPlace} and Proposition~\ref{SuffCondCircuit}.

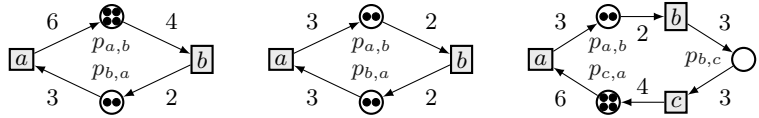
\begin{figure}[!ht]
\centering

\begin{tikzpicture}[scale=0.6]
\node[place,tokens=4](p)at(2,0.95)[label=below:$p_{a,b}$]{};
\node[place,tokens=2](p')at(2,-0.95)[label=above:$p_{b,a}$]{};
\node[transition](a)at(0,0){$a$};
\node[transition](b)at(4,0){$b$};
\draw(a)[]edge[-latex]node[above left]{$6$}(p);
\draw(p)[]edge[-latex]node[above right]{$4$}(b);
\draw(p')[]edge[-latex]node[below left]{$3$}(a);
\draw(b)[]edge[-latex]node[below right]{$2$}(p');
\end{tikzpicture}
\hspace*{4mm}
\raisebox{0mm}{
\begin{tikzpicture}[scale=0.6]
\node[place,tokens=2](p)at(2,0.95)[label=below:$p_{a,b}$]{};
\node[place,tokens=2](p')at(2,-0.95)[label=above:$p_{b,a}$]{};
\node[transition](a)at(0,0){$a$};
\node[transition](b)at(4,0){$b$};
\draw(a)[]edge[-latex]node[above left]{$3$}(p);
\draw(p)[]edge[-latex]node[above right]{$2$}(b);
\draw(p')[]edge[-latex]node[below left]{$3$}(a);
\draw(b)[]edge[-latex]node[below right]{$2$}(p');
\end{tikzpicture}
}
\hspace*{4mm}
\begin{tikzpicture}[scale=0.6]
\node[place,tokens=2](pab)at(1.5,0.95)[label=below:$p_{a,b}$]{};
\node[place,tokens=0](pbc)at(4.5,0)[label=left:$p_{b,c}$]{};
\node[place,tokens=4](pca)at(1.5,-0.95)[label=above:$p_{c,a}$]{};
\node[transition](a)at(0,0){$a$};
\node[transition](b)at(3,0.95){$b$};
\node[transition](c)at(3,-0.95){$c$};
\draw(a)[]edge[-latex]node[above left]{$3$}(pab);
\draw(pab)[]edge[-latex]node[below]{$2$}(b);
\draw(b)[]edge[-latex]node[above right]{$3$}(pbc);
\draw(pbc)[]edge[-latex]node[below right]{$3$}(c);
\draw(c)[]edge[-latex]node[above]{$4$}(pca);
\draw(pca)[]edge[-latex]node[below left]{$6$}(a);
\end{tikzpicture}
\vspace*{-2mm}
\caption{These circuits are conservative and fulfill the sufficient condition of liveness of Proposition~\ref{SuffCondCircuit}.
On the left, $\gcd_{p_{a,b}}=2$ and $\gcd_{p_{b,a}}=1$.
The system in the middle is obtained from the one on the left by scaling $p_{a,b}$ with $\frac{1}{2}$.
In this second system, $\gcd_{p_{a,b}}=1$ and $\gcd_{p_{b,a}}=1$.
On the right, $\gcd_{p_{a,b}}=1$, $\gcd_{p_{b,c}}=3$ and $\gcd_{p_{c,a}}=2$. 
}
\label{solcyc2duplicate3.fig} 
\end{figure}

We recall a liveness characterisation.
Since a place with an output and no input, called a {\em source-place}, prevents liveness, we assume there is no such place.

\begin{proposition}[Liveness of WMGs~\cite{WTS92}]\label{CharLiveWMG}
Consider a WMG $\mathcal{S}$ without source places.
Then $\mathcal{S}$ is live iff each circuit P-subsystem of $\mathcal{S}$ is live.
\end{proposition}

\subsection{Weak synthesis of WMGs in polynomial-time}

Algorithm~\ref{algoWMGptime} below constructs a WMG from a given prime T-vector $\uniqueP$.
We prove it terminates and computes a WMG cyclically solving some word
with Parikh vector $\uniqueP$, hence performing weak synthesis.
We then show it lies in PTIME.
\begin{algorithm}[h]\label{algoWMGptime}
\KwData{A prime T-vector $\uniqueP$ with support $T=\{t_1, \ldots, t_m\}$.}
\KwResult{A WMG cyclically solving a word with Parikh vector $\uniqueP$.}

We construct first an unmarked WMG $N=(P,T,W)$ containing all possible binary circuits (which we call the complete WMG), as follows:\\
\For{each pair of distinct labels $t_i, t_j$ in $T$}{ 
Add two new places $p_{i,j}$ and $p_{j,i}$ forming a binary circuit P-subnet with set of labels $\{t_i,t_j\}$, 
such that:\\
$W(p_{i,j},t_j)=W(t_j,p_{j,i})=\frac{\uniqueP(t_i)}{\gcd(\uniqueP(t_i),\uniqueP(t_j))}$\\
$W(p_{j,i},t_i)=W(t_i,p_{i,j})=\frac{\uniqueP(t_j)}{\gcd(\uniqueP(t_i),\uniqueP(t_j))}$.\\
}
Then, we construct its initial marking $M_0$, visiting the transitions in increasing order, as follows:\\ 
\For{$i=2..m$}{
Mark each output place $p_{i,h}$ of $t_i$ that is an input of a transition $t_h$ of smaller index, i.e. $h < i$, 
with $M_0(p) = W(p_{i,h},t_h)$;\\
Mark each input place $p_{h,i}$ of $t_i$ that is an output of a transition $t_h$ of smaller index, i.e. $h < i$, 
with $M_0(p) = W(p_{h,i},t_i) - \gcd_{p_{h,i}} = W(p_{h,i},t_i) - 1$;
}

\Return{$(N,M_0)$}

\caption{Weak synthesis of a WMG with circular RG.}
\end{algorithm}

\begin{theorem}[Weak synthesis of a WMG]\label{CyclicWMGprime}
For every prime T-vector $\uniqueP$, 
Algorithm~\ref{algoWMGptime} terminates and computes a WMG cyclically solving $\uniqueP$,
i.e. 
cyclically solving some word $w \in T^\ast$ such that $\Parikh(w) = \uniqueP$. 
\end{theorem}
\begin{proof}
The proof is illustrated in Fig.~\ref{proofInductionCyclicWMGsingleAndBinary.fig}, \ref{proofInductionCyclicWMG.fig} and~\ref{proofInductionCyclicWMGbis.fig}.
Consider any prime T-vector $\uniqueP \in (\mathbb{N} \setminus \{0\})^m$, where $m$ is the number of transitions.
In the first loop, we consider each pair of transitions once.
In the second loop, we consider each place once. Thus, the algorithm terminates.
Let us prove its correction.

If $|T| = 1$, there is one transition and $\uniqueP=(1)$: the WMG with $T = \{t_1\}$, $P = \emptyset$
and the sequence $w=t_1$ fulfill the claim.
Hence, we suppose $|T| \ge 2$.

For each place $p_{i,j}$, we have $\uniqueP(t_j) \cdot W(p_{i,j},t_j) = \uniqueP(t_i) \cdot W(t_i, p_{i,j})$,
so that
$- W(p_{i,j},t_j) \cdot \uniqueP(t_j) + W(t_i, p_{i,j}) \cdot \uniqueP(t_i) = 0$,
hence
$I \cdot \uniqueP = 0$, where $I$ is the incidence matrix of $N$.
Moreover, each circuit P-subnet of $N$ is conservative (by Corollary 3.6 in~\cite{WTS92}).

Now, let us consider the second loop: we prove the next invariant $\mathrm{Inv}(\ell)$ to be true at the end of each iteration $\ell$,
for each $\ell=1..m-1$, by induction on $\ell$:

\noindent $\mathrm{Inv}(\ell)$: 
"At the end of the $\ell$-th iteration, the WMG P-subsystem $\mathcal{S}_{\ell}$ 
defined by the set of places $P_\ell = \{p_{u,v} \mid u,v \in \{1, \ldots, \ell+1\}, u \neq v\}$ is live,
and each binary circuit P-subsystem of $\mathcal{S}_{\ell}$ has exactly one enabled place".

Before entering the loop, i.e. before the first iteration, the WMG in unmarked.

Base case: $\ell=1$. At the end of the first iteration, $P_\ell = P_1 = \{p_{1,2}, p_{2,1}\}$, 
which induces a live binary circuit (by Proposition~\ref{SuffCondCircuit})
with exactly one enabled place, since only one output of $t_2$ is enabled by $M_0$
and the other place is an output of $t_1$ considered in the second part of the loop. 

Inductive case: $1 < \ell \le m-1$. We suppose $\mathrm{Inv}(\ell-1)$ to be true, and we prove that $\mathrm{Inv}(\ell)$ is true.
Thus, at the end of iteration $\ell-1$, we suppose that the P-subsystem $\mathcal{S}_{\ell-1}$ induced by $P_{\ell-1}$ is live,
and that each binary circuit P-subsystem of $\mathcal{S}_{\ell-1}$ has exactly one enabled place.
The iteration $\ell$ marks only all the input and output places of $t_{\ell+1}$ 
that are inputs or outputs of transitions in $\{t_1, \ldots, t_\ell\}$.
None of these places has been considered in any previous iteration, 
since each iteration considers only places connected to transitions of smaller index.
Thus, these places are newly marked at iteration $\ell$,
and the only places unmarked at the end of this iteration are connected to transitions of higher index.

We deduce that, at the end of iteration $\ell$:\\
$-$ each binary circuit of $\mathcal{S}_\ell$ has exactly one enabled place:
indeed, each such binary circuit either belongs to $\mathcal{S}_{\ell-1}$, on which the inductive hypothesis applies,
or to the circuits newly marked at iteration $\ell$;\\
$-$ each circuit P-subsystem of $\mathcal{S}_\ell$ with three places or more is live:
indeed, consider any such conservative circuit $C$; 
either $C$ is a P-subsystem of $\mathcal{S}_{\ell-1}$, which is live by the inductive hypothesis,
hence Proposition~\ref{CharLiveWMG} applies and $C$ is live,
or $C$ contains transition $t_{\ell+1}$,
in which case $C$ contains necessarily an output $p_{\ell+1,h}$ of $t_{\ell+1}$ with $h < \ell+1$: 
since each place $p_{u,v}$ of $\mathcal{S}_\ell$ is marked with at least $W(p_{u,v},t_v) - \gcd_{p_{u,v}}$
and
$p_{\ell+1,h}$ is marked with $W(p_{\ell+1,h},t_h)$, 
$C$ fulfills the sufficient condition of liveness of Proposition~\ref{SuffCondCircuit}, hence is live;
we deduce that each circuit P-subsystem is live, hence $\mathcal{S}_\ell$ is live by Proposition~\ref{CharLiveWMG}. 

We proved that $\mathrm{Inv}(\ell)$ is true for every integer $\ell=1..m-1$.
We deduce that the WMG system $\mathcal{S}_{m-1}=(N,M_0)$ obtained at the end of the last iteration, which is the system returned,
fulfills $\mathrm{Inv}(m-1)$.
Suppose that some marking $M$ reachable in $\mathcal{S}_{n-1}$ enables two distinct transitions $t_i$ and $t_j$.
Since $\mathcal{S}_{n-1}$ is a complete WMG, 
there is a binary circuit P-subsystem $C_{i,j}=(N_{i,j},M_{i,j})$ containing $t_i$ and $t_j$,
in which exactly one place is enabled, applying Lemma~\ref{OneEnabledPlace} 
(since $\projection{M}{\{p_{i,j},p_{j,i}\}} = M_{i,j}$ is a marking reachable in $(N_{i,j},\projection{M_0}{\{p_{i,j},p_{j,i}\}})$).
We deduce that $M$ cannot enable both $t_i$ and $t_j$, a contradiction.

Thus the WMG returned is live and each of its reachable markings enables exactly one transition.
It is known that, in each live and bounded WMG, a sequence $\sigma$ is feasible 
such that $\Parikh(\sigma) = \uniqueP$ which is the unique minimal T-semiflow of the WMG (see \cite{WTS92,tcs97}).
Consequently, its reachability graph is a circle, i.e. the WMG solves $\uniqueP$ (and $\sigma$) cyclically.
We get the claim.
\end{proof}

\begin{figure}[!ht]
\centering

\raisebox{5.8mm}{
\begin{tikzpicture}[scale=0.6]
\node[transition](t1)at(0,0){$t_1$};
\end{tikzpicture}
}
\hspace*{1cm}
\begin{tikzpicture}[scale=0.6]
\node[place,fill=black!15](p12)at(2,0.95)[label=below:$p_{1,2}$]{};
\node[place,fill=black](p21)at(2,-0.95)[label=above:$p_{2,1}$]{};
\node[transition](t1)at(0,0){$t_1$};
\node[transition](t2)at(4,0){$t_2$};
\draw(t1)[]edge[-latex,bend left=0]node[above left]{$w_{1,2}$}(p12);
\draw(p12)[]edge[-latex,bend left=0]node[above right]{$w_{2,1}$}(t2);
\draw(t2)[]edge[-latex]node[below right]{$w_{2,1}$}(p21);
\draw(p21)[]edge[-latex]node[below left]{$w_{1,2}$}(t1);
\end{tikzpicture}
\hspace*{1cm}
\begin{tikzpicture}[scale=0.6]
\node[place,tokens=2](p12)at(2,0.95)[label=below:$p_{1,2}$]{};
\node[place,tokens=2](p21)at(2,-0.95)[label=above:$p_{2,1}$]{};
\node[transition](t1)at(0,0){$t_1$};
\node[transition](t2)at(4,0){$t_2$};
\draw(t1)[]edge[-latex,bend left=0]node[above left]{$2$}(p12);
\draw(p12)[]edge[-latex,bend left=0]node[above right]{$3$}(t2);
\draw(t2)[]edge[-latex]node[below right]{$3$}(p21);
\draw(p21)[]edge[-latex]node[below left]{$2$}(t1);
\end{tikzpicture}
\vspace*{-1mm}
\caption{Sketching Theorem~\ref{CyclicWMGprime} for $1$ and $2$ transitions.
On the left, the circuit system $\mathcal{S}_1$ has no place and solves $\uniqueP=(1)$.
In the circuit system $\mathcal{S}_2$ in the middle, the output of $t_2$ is marked as black and its input as grey.
On the right, an instanciation of the binary case, given $\uniqueP=(3,2)$.
These systems are live, $RG(\mathcal{S}_1)$ and $RG(\mathcal{S}_2)$ are circles.
}
\label{proofInductionCyclicWMGsingleAndBinary.fig}
\end{figure}
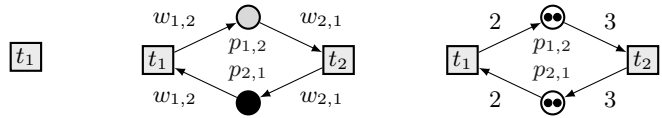

\begin{figure}[!ht]
\centering

\begin{tikzpicture}[scale=0.85]
\node[place,fill=black!15](p12)at(2,2.1)[label=below:$p_{1,2}$]{};
\node[place,fill=black](p21)at(2,1.1)[label=below:$p_{2,1}$]{};
\node[place,fill=black!15](p13)at(0,-2.3)[label=below:$p_{1,3}$]{};
\node[place,fill=black](p31)at(-0.9,-2.9)[label=below:$p_{3,1}$]{};
\node[place,fill=black!15](p23)at(4.9,-2.9)[label=below:$p_{2,3}$]{};
\node[place,fill=black](p32)at(4,-2.3)[label=below:$p_{3,2}$]{};
\node[transition](t1)at(0,0){$t_1$};
\node[transition](t2)at(4,0){$t_2$};
\node[transition](t3)at(2,-3){$t_3$};
\draw(t1)[]edge[-latex,bend left=20]node[above left]{$w_{1,2}$}(p12);
\draw(p12)[]edge[-latex,bend left=20]node[above right]{$w_{2,1}$}(t2);
\draw(t2)[]edge[-latex]node[above]{$w_{2,1}$}(p21);
\draw(p21)[]edge[-latex]node[above]{$w_{1,2}$}(t1);
\draw(t1)[]edge[-latex]node[below left,near end]{$w_{1,3}$}(p13);
\draw(p13)[]edge[-latex]node[below]{$w_{3,1}$}(t3);
\draw(t3)[]edge[-latex]node[below]{$w_{3,2}$}(p32);
\draw(p32)[]edge[-latex]node[below right,near start]{$w_{2,3}$}(t2);
\draw(t3)[]edge[-latex,bend left=20]node[below]{$w_{3,1}$}(p31);
\draw(p31)[]edge[-latex,bend left=20]node[above left, very near end]{$w_{1,3}$}(t1);
\draw(t2)[]edge[-latex,bend left=20]node[above right, very near start]{$w_{2,3}$}(p23);
\draw(p23)[]edge[-latex,bend left=20]node[below]{$w_{3,2}$}(t3);
\end{tikzpicture}
\hspace*{3mm}
\begin{tikzpicture}[scale=0.85]
\node[place,fill=black!15](p12)at(2,2.1)[label=below:$p_{1,2}$]{};
\node[place,fill=black](p21)at(2,1.1)[label=right:$p_{2,1}$]{};
\node[place,fill=black!15](p13)at(0,-2.3)[label=below:$p_{1,3}$]{};
\node[place,fill=black](p31)at(-0.9,-2.9)[label=below:$p_{3,1}$]{};
\node[place,fill=black!15](p23)at(4.9,-2.9)[label=below:$p_{2,3}$]{};
\node[place,fill=black](p32)at(4,-2.3)[label=below:$p_{3,2}$]{};
\node[place,fill=black!15](p14)at(0.3,-1.2)[label=below:~$p_{1,4}$]{};
\node[place,fill=black](p41)at(1.65,0)[label=above:$p_{4,1}$~]{};
\node[place,fill=black!15](p24)at(3.7,-1.2)[label=below:$p_{2,4}$]{};
\node[place,fill=black](p42)at(2.35,0)[label=above:$p_{4,2}$]{};
\node[place,fill=black!15](p34)at(2.4,-2.1)[label=right:$p_{3,4}$]{};
\node[place,fill=black](p43)at(1.6,-2.1)[label=left:$p_{4,3}$]{};
\node[transition](t1)at(0,0){$t_1$};
\node[transition](t2)at(4,0){$t_2$};
\node[transition](t3)at(2,-3){$t_3$};
\node[transition](t4)at(2,-1.2){$t_4$};
\draw(t1)[]edge[-latex,bend left=25]node[above left]{$w_{1,2}$}(p12);
\draw(p12)[]edge[-latex,bend left=25]node[above right]{$w_{2,1}$}(t2);
\draw(t2)[]edge[-latex]node[above,near start]{$w_{2,1}$}(p21);
\draw(p21)[]edge[-latex]node[above,near end]{$w_{1,2}$}(t1);
\draw(t1)[]edge[-latex]node[below left,near end]{$w_{1,3}$}(p13);
\draw(p13)[]edge[-latex]node[below]{$w_{3,1}$}(t3);
\draw(t3)[]edge[-latex]node[below]{$w_{3,2}$}(p32);
\draw(p32)[]edge[-latex]node[below right,near start]{$w_{2,3}$}(t2);
\draw(t3)[]edge[-latex,bend left=20]node[below]{$w_{3,1}$}(p31);
\draw(p31)[]edge[-latex,bend left=20]node[above left, very near end]{$w_{1,3}$}(t1);
\draw(t2)[]edge[-latex,bend left=20]node[above right, very near start]{$w_{2,3}$}(p23);
\draw(p23)[]edge[-latex,bend left=20]node[below]{$w_{3,2}$}(t3);
\draw(t3)[]edge[-latex,bend left=0]node[right]{$w_{3,4}$}(p34);
\draw(p34)[]edge[-latex,bend left=0]node[right]{$w_{4,3}$}(t4);
\draw(t4)[]edge[-latex,bend left=0]node[left]{$w_{4,3}$}(p43);
\draw(p43)[]edge[-latex,bend left=0]node[left]{$w_{3,4}$}(t3);
\draw(t4)[]edge[-latex,bend left=0]node[left]{$w_{4,1}$}(p41);
\draw(p41)[]edge[-latex,bend left=0]node[below]{$w_{1,4}$}(t1);
\draw(t4)[]edge[-latex,bend left=0]node[right]{$w_{4,2}$}(p42);
\draw(p42)[]edge[-latex,bend left=0]node[below]{$w_{2,4}$}(t2);
\draw(t1)[]edge[-latex,bend left=0]node[right]{$w_{1,4}$}(p14);
\draw(p14)[]edge[-latex,bend left=0]node[above]{$w_{4,1}$}(t4);
\draw(t2)[]edge[-latex,bend left=0]node[left]{$w_{2,4}$}(p24);
\draw(p24)[]edge[-latex,bend left=0]node[above]{$w_{4,2}$}(t4);
\end{tikzpicture}
\vspace*{-2mm}
\caption{Sketching Theorem~\ref{CyclicWMGprime} for $3$ and $4$ transitions (systems $\mathcal{S}_3$, $\mathcal{S}_4$).
Each black place $p_{i,j}$ is marked with $W(p_{i,j},t_j)$, each grey place $p_{i,j}$ is marked with $W(p_{i,j},t_j) - \gcd_{p_{i,j}}$.
On the left, in the circuit induced by $\{p_{1,2},p_{2,1}\}$, the output of $t_2$ is black and its input is grey.
Then, each output of $t_3$ is black, each input is grey.
Each circuit of $\mysyst_3$ is live, $\mysyst_3$ is live and $RG(\mysyst_3)$ is a circle.
In $\mysyst_4$, we keep the marking of $\mysyst_3$ and mark each output of $t_4$ as black, each of its inputs as grey.
Thus, $\mysyst_4$ is live and $RG(\mysyst_4)$ is a circle.
}
\label{proofInductionCyclicWMG.fig}
\end{figure}
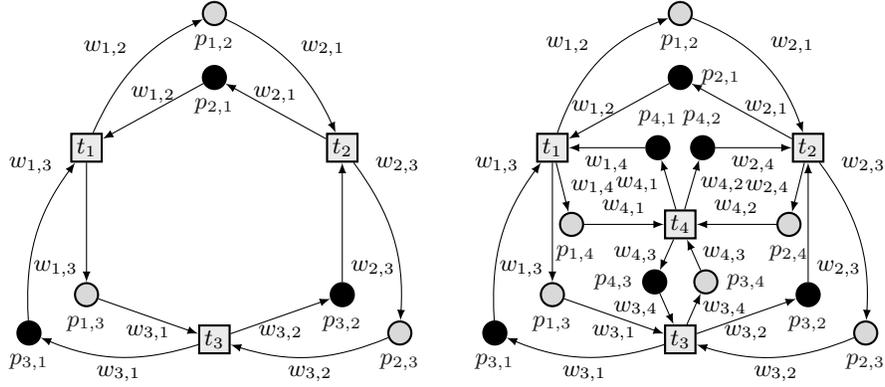

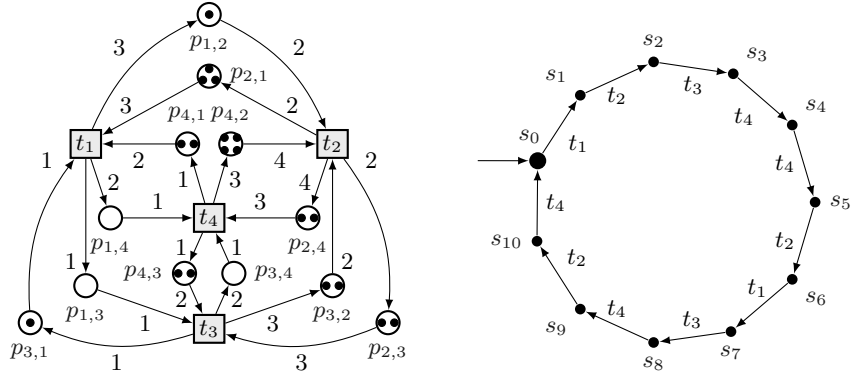
\begin{figure}[!ht]
\centering
\begin{tikzpicture}[scale=0.82]
\node[place,tokens=1](p12)at(2,2.1)[label=below:$p_{1,2}$]{};
\node[place,tokens=3](p21)at(2,1.1)[label=right:$p_{2,1}$]{};
\node[place](p13)at(0,-2.3)[label=below:$p_{1,3}$]{};
\node[place,tokens=1](p31)at(-0.9,-2.9)[label=below:$p_{3,1}$]{};
\node[place,tokens=2](p23)at(4.9,-2.9)[label=below:$p_{2,3}$]{};
\node[place,tokens=2](p32)at(4,-2.3)[label=below:$p_{3,2}$]{};
\node[place](p14)at(0.4,-1.2)[label=below:$p_{1,4}$]{};
\node[place,tokens=2](p41)at(1.65,0)[label=above:$p_{4,1}$]{};
\node[place,tokens=2](p24)at(3.6,-1.2)[label=below:$p_{2,4}$]{};
\node[place,tokens=4](p42)at(2.35,0)[label=above:$p_{4,2}$]{};
\node[place](p34)at(2.4,-2.1)[label=right:$p_{3,4}$]{};
\node[place,tokens=2](p43)at(1.6,-2.1)[label=left:$p_{4,3}$]{};
\node[transition](t1)at(0,0){$t_1$};
\node[transition](t2)at(4,0){$t_2$};
\node[transition](t3)at(2,-3){$t_3$};
\node[transition](t4)at(2,-1.2){$t_4$};
\draw(t1)[]edge[-latex,bend left=20]node[above left]{$3$}(p12);
\draw(p12)[]edge[-latex,bend left=20]node[above right]{$2$}(t2);
\draw(t2)[]edge[-latex]node[above,near start]{$2$}(p21);
\draw(p21)[]edge[-latex]node[above, near end]{$3$}(t1);
\draw(t1)[]edge[-latex]node[below left,near end]{$1$}(p13);
\draw(p13)[]edge[-latex]node[below]{$1$}(t3);
\draw(t3)[]edge[-latex]node[below]{$3$}(p32);
\draw(p32)[]edge[-latex]node[below right,near start]{$2$}(t2);
\draw(t3)[]edge[-latex,bend left=20]node[below]{$1$}(p31);
\draw(p31)[]edge[-latex,bend left=20]node[above left, very near end]{$1$}(t1);
\draw(t2)[]edge[-latex,bend left=20]node[above right, very near start]{$2$}(p23);
\draw(p23)[]edge[-latex,bend left=20]node[below]{$3$}(t3);
\draw(t3)[]edge[-latex,bend left=0]node[right]{$2$}(p34);
\draw(p34)[]edge[-latex,bend left=0]node[right]{$1$}(t4);
\draw(t4)[]edge[-latex,bend left=0]node[left]{$1$}(p43);
\draw(p43)[]edge[-latex,bend left=0]node[left]{$2$}(t3);
\draw(t4)[]edge[-latex,bend left=0]node[left]{$1$}(p41);
\draw(p41)[]edge[-latex,bend left=0]node[below]{$2$}(t1);
\draw(t4)[]edge[-latex,bend left=0]node[right]{$3$}(p42);
\draw(p42)[]edge[-latex,bend left=0]node[below]{$4$}(t2);
\draw(t1)[]edge[-latex,bend left=0]node[right]{$2$}(p14);
\draw(p14)[]edge[-latex,bend left=0]node[above]{$1$}(t4);
\draw(t2)[]edge[-latex,bend left=0]node[left]{$4$}(p24);
\draw(p24)[]edge[-latex,bend left=0]node[above]{$3$}(t4);
\end{tikzpicture}
\hspace*{6mm}
\begin{tikzpicture}[scale=1.3]
\begin{scope}[xshift=0cm,yshift=0cm]
\node[circle,fill=black!100,inner sep=0.08cm,label=above:$s_0~~$](00)at(163:1.45)[]{};
\node[circle,fill=black!100,inner sep=0.05cm,label=above left:$s_1$](01)at(131:1.45)[]{};
\node[circle,fill=black!100,inner sep=0.05cm,label=above:$s_2$](02)at(98:1.45)[]{};
\node[circle,fill=black!100,inner sep=0.05cm,label=above right:$s_3$](03)at(65:1.45)[]{};
\node[circle,fill=black!100,inner sep=0.05cm,label=above right:$s_4$](04)at(33:1.45)[]{};
\node[circle,fill=black!100,inner sep=0.05cm,label=right:$s_5$](05)at(0:1.45)[]{};
\node[circle,fill=black!100,inner sep=0.05cm,label=below right:$s_6$](06)at(327:1.45)[]{};
\node[circle,fill=black!100,inner sep=0.05cm,label=below:$s_7$](07)at(294:1.45)[]{};
\node[circle,fill=black!100,inner sep=0.05cm,label=below:$s_8$](08)at(262:1.45)[]{};
\node[circle,fill=black!100,inner sep=0.05cm,label=below left:$s_9$](09)at(229:1.45)[]{};
\node[circle,fill=black!100,inner sep=0.05cm,label=left:$s_{10}$](10)at(196:1.45)[]{};

\draw[-latex](-2,0.425)to node[]{}(00);
\draw[-latex](00)to node[auto,below right]{$t_1$}(01);
\draw[-latex](01)to node[auto,below]{$t_2$}(02);
\draw[-latex](02)to node[auto,below]{$t_3$}(03);
\draw[-latex](03)to node[auto,below left]{$t_4$}(04);
\draw[-latex](04)to node[auto,left]{$t_4$}(05);
\draw[-latex](05)to node[auto,left]{$t_2$}(06);
\draw[-latex](06)to node[auto,above]{$t_1~$}(07);
\draw[-latex](07)to node[auto,above]{$t_3$}(08);
\draw[-latex](08)to node[auto,above]{$t_4$}(09);
\draw[-latex](09)to node[auto,above right]{$t_2$}(10);
\draw[-latex](10)to node[auto,right]{$t_4$}(00);
\end{scope}
\end{tikzpicture}
\vspace*{-1mm}
\caption{Illustration of the proof of Theorem~\ref{CyclicWMGprime} for the prime T-vector $\uniqueP = (2,3,2,4)$.
On the left, a complete WMG $\mysyst$, with all possible binary circuits.
Its marking follows the black and grey places of Fig.~\ref{proofInductionCyclicWMG.fig}.
Pick any circuit P-subsystem $C$, e.g. the one induced by $\{p_{4,3}, p_{3,2}, p_{2,4}\}$:
it is conservative and fulfills the condition of Proposition~\ref{SuffCondCircuit}, hence it is live.
On the right, an \lts{} representing $RG(\mysyst)$.
The sequence $w=t_1 \, t_2 \, t_3 \, t_4 \, t_4 \, t_2 \, t_1 \, t_3 \, t_4 \, t_2 \, t_4$,
with $\Parikh(w) = \uniqueP$, is cyclically WMG-solvable.
}
\label{proofInductionCyclicWMGbis.fig}
\end{figure}

\vspace*{\baselineskip}
\noindent {\bf Polynomial-time complexity of Algorithm~\ref{algoWMGptime}.}
Let $m$ be the number of transitions (labels). 
The initial construction of the net $N$ considers a number of transition pairs equal to $m \cdot (m-1)$.
The computation of $\gcd(i,j)$ for any two integers $i,j$ can be done using the Euclidean algorithm in $\mathcal{O}(log_2^2(\max(i,j)))$,
which remains polynomial in the size of the input vector binary encoding.
Computing $\frac{u}{v}$ for two integers $u,v$, $v \neq 0$, can also be done in $\mathcal{O}(log_2^2(\max(u,v)))$.
Thus, constructing $N$ lies in $\mathcal{O}(m (m-1) 3 log_2^2(q))$ where $q$ is the highest value in $\uniqueP$.
Then, the algorithm marks all places in $\mathcal{O}(m (m-1))$, knowing that the $\gcd$ of each place is~$1$.
Hence the algorithm lies in PTIME: $\mathcal{O}(3 m (m-1) log_2^2(q) + m (m-1))$,
i.e.
$\mathcal{O}(m (m-1) (3 log_2^2(q) + 1))$ where $q$ is the highest value in $\uniqueP$.\\

\noindent {\bf Comparison with sequence-based synthesis.}
Algorithm~\ref{wmg.alg} uses $\mathcal{O}(n(m+n))$ steps, 
where 
$m$ is the number of labels and $n$ is the length of the input sequence $w$ to be solved cyclically.
Since $n$ equals the sum of the components of $\uniqueP = \Parikh(w)$, we get $q \le n$.
Also, $m \le n$; depending on the weights, $n$ can be exponentially larger than $m$.
Hence, $log_2^2(q) \in \mathcal{O}(log_2^2(n))$, so that $m (m-1) (3 log_2^2(q) + 1) \in \mathcal{O}(m^2 \cdot log_2^2(n))$. 
When $n$ is exponential in $m$, Algorithm~\ref{algoWMGptime} operates in time polynomial in $m$
while Algorithm~\ref{wmg.alg} operates in time exponential in $m$.

\section{Conclusions and Perspectives}\label{conclu.sec} 

In this work, we specialised previous methods of analysis and synthesis to the CF nets
and their WMG subclass, two useful subclasses of weighted Petri nets
allowing to model various real-world applications.

We highlighted the correspondance between CF- and WMG-solvability for binary alphabets. 
We also tackled the case of an \lts{} formed of a single cycle with an arbitrary number of letters,
for which we developed a characterisation of WMG-solvability together with a dedicated polynomial-time synthesis algorithm. 
We showed the equivalence between cyclic WMG- and CF-solvability in the case of three-letter alphabets,
and that it does not extend to four-letter alphabets.
We also discussed the applicability of our conditions to cyclic CF synthesis over arbitrary alphabets.

Finally, we introduced the notion of weak synthesis, allowing to be less restrictive on the solution design,
and provided a polynomial-time algorithm weakly synthesising a WMG with circular reachability graph. 
We showed this second algorithm to often operate much faster than the sequence-based one.

As a natural continuation of the work, we expect extensions of our results in two directions: 
generalising the class of goal-nets 
and relaxing the restrictions on the \lts{} under consideration.

%

\bibliographystyle{splncs}
\bibliography{TOPNOC-2019-2020}

\end{document}